\pdfoutput=1

\pdfoutput=1

\documentclass[11pt]{scrartcl}
\usepackage[a4paper, total={16cm, 24cm}]{geometry}
\usepackage{microtype} 
\usepackage{authblk}

\usepackage{xcolor}
\usepackage{graphicx}
\usepackage{pgfplots}
\pgfplotsset{compat=1.17}

\usepackage{amsmath, amssymb, amsthm, mathtools, nicefrac, siunitx}

\usepackage{thmtools}
\usepackage{thm-restate}

\newtheorem{theorem}{Theorem}[section]
\newtheorem{claim}[theorem]{Claim}
\newtheorem{example}[theorem]{Example}
\newtheorem{observation}[theorem]{Observation}
\newtheorem{proposition}[theorem]{Proposition}

\theoremstyle{definition}
\newtheorem{definition}[theorem]{Definition}

\usepackage[pagebackref]{hyperref}
\hypersetup{
	pdfencoding=auto,
	psdextra,
	colorlinks=true,
	citecolor=green!50!black,
	linkcolor=red!60!black,
}

\usepackage[nameinlink]{cleveref}

\usepackage{natbib}
\usepackage{algorithm}
\usepackage[noend]{algpseudocode}

\usepackage[inline]{enumitem}
\usepackage{booktabs}
\usepackage{subcaption}

\newcommand{\restatehere}[1]{%
	\marginline{\vspace{0.6cm}\footnotesize \hyperlink{original#1}{\hypertarget{restated#1}{[Main]}}}%
	\csname #1\endcsname*%
}

\usepackage{tabularx}

\Crefname{theorem}{theorem}{theorems}
\Crefname{theorem}{Theorem}{Theorems}

\Crefname{claim}{claim}{claims}
\Crefname{claim}{Claim}{Claims}

\Crefname{example}{example}{examples}
\Crefname{example}{Example}{Examples}

\Crefname{observation}{observation}{observations}
\Crefname{observation}{Observation}{Observations}

\Crefname{proposition}{proposition}{propositions}
\Crefname{proposition}{Proposition}{Propositions}

\Crefname{lemma}{lemma}{lemmas}
\Crefname{lemma}{Lemma}{Lemmas}

\Crefname{corollary}{corollary}{corollaries}
\Crefname{corollary}{Corollary}{Corollaries}

\Crefname{definition}{definition}{definitions}
\Crefname{definition}{Definition}{Definitions}

\Crefname{remark}{remark}{remarks}
\Crefname{remark}{Remark}{Remarks}

\usepackage{footnote}
\makesavenoteenv{algorithm}

\usepackage{caption}
\usepackage{subcaption}

\usepackage[inline]{enumitem}

\usepackage{url}            
\usepackage{booktabs}       
\usepackage{microtype}      
\usepackage[utf8]{inputenc} 

\newcommand{\NN}{\mathbb{N}} 
\newcommand{\RR}{\mathbb{R}} 
\newcommand{\RRR}{\mathbb{R}_{\geq 0}}

\newcommand{\bigO}{\mathcal{O}} 

\newcommand{\EFOEX}{\textsc{EF1-init Existence}}
\newcommand{\EFEX}{\textsc{EF-init Existence}}

\newcommand{\problemdef}[3]{
  \begin{center}
    \begin{minipage}{0.95\textwidth}
      \noindent
      \textsc{#1}

      \vspace{2pt}
      \setlength{\tabcolsep}{3pt}
      \begin{tabularx}{\textwidth}{@{}lX@{}}
        \textbf{Input:}     & #2 \\
        \textbf{Question:}   & #3
      \end{tabularx}
    \end{minipage}
  \end{center}
}

\newenvironment{claimproof}[1][\proofname]
{%
  \proof[#1]%
}
{%
  \endproof%
}

\DeclarePairedDelimiter\floor{\lfloor}{\rfloor}

\makeatletter
\newcommand{\alglineref}[1]{Line~\ref{#1}}
\makeatother

\makeatletter
  \newcommand{\LongState}[2][0pt]{%
    \State \hspace{#1}\parbox[t]{\dimexpr\linewidth-\ALG@thistlm}{\hangindent=0.25em \raggedright #2}%
    \vspace{0.3em}%
  }
  \newcommand{\LongStatex}[2][0pt]{%
    \Statex \hspace{#1}\parbox[t]{\dimexpr\linewidth-\ALG@thistlm-#1}{\hangindent=0.25em \raggedright #2}%
    \vspace{0.3em}%
  }
  \newcommand{\LongElsIfState}[1]{%

  \State\parbox[t]{\dimexpr\linewidth-\ALG@thistlm}{\hangindent=0.25em \raggedright \textbf{else if} #1 \textbf{then}}%
    \vspace{0.3em}%
  }
  \makeatother

\title{Fair Allocation with Initial Utilities\footnote{An extended abstract of the paper is scheduled to appear in the proceedings of the  \emph{25th International Conference on Autonomous Agents and Multiagent Systems
(AAMAS 2026)}. }}

\date{\vspace{-1.5cm}}

\author[1]{Niclas Boehmer}
\author[1]{Luca Kreisel}
\affil[1]{Hasso Plattner Institute, University of Potsdam, Germany} 
\affil[ ]{\texttt{Niclas.Boehmer@hpi.de,Luca.Kreisel@hpi.de}} 

\begin{document}

\maketitle

\bigskip
{\footnotesize\tableofcontents}

\newpage 
\begin{abstract}
	\begin{center}
		\textbf{\textsf{Abstract}} \smallskip
	\end{center}
The problem of allocating indivisible resources to agents arises in a wide range of domains, including treatment distribution and social support programs. An important goal in algorithm design for this problem is fairness, where the focus in previous work has been on ensuring that the computed allocation provides equal treatment to everyone. However, this perspective disregards that agents may start from unequal initial positions, which is crucial to consider in settings where fairness is understood as \emph{equality of outcome}. In such settings, the goal is to create an equal final outcome for everyone by leveling initial inequalities through the allocated resources. To close this gap, focusing on agents with additive utilities, we extend the classic model by assigning each agent an initial utility and study the existence and computational complexity of several new fairness notions following the principle of equality of outcome. Among others, we show that complete allocations satisfying a direct analog of envy-freeness up to one resource (EF1) may fail to exist and are computationally hard to find, forming a contrast to the classic setting without initial utilities. We propose a new, always satisfiable fairness notion, called minimum-EF1-init and design a polynomial-time algorithm based on an extended round-robin procedure to compute complete allocations satisfying this notion.	
	\end{abstract}

\section{Introduction}

Allocating resources to agents is a central problem in algorithmic decision-making, with applications ranging from task and house allocation to treatment distribution and scheduling \citep{DBLP:conf/icml/0002B023,moulin2019fair,DBLP:journals/sigecom/GoldmanP14,bastani2021efficient,moulin2004fair,DBLP:conf/aaai/AleksandrovW20,emanuel2023shared}. Thus, the problem has been widely studied across artificial intelligence, machine learning, multi-agent systems, operations research, and computational social choice. Much of this literature focuses on the design and analysis of algorithms that compute allocations satisfying formal fairness criteria, for instance, envy-freeness \citep{DBLP:journals/ai/AmanatidisABFLMVW23,DBLP:reference/choice/BouveretCM16,DBLP:journals/jair/LiuLSW24}.

Both outside and inside computer science, the discourse on fairness distinguishes between two paradigms:
On the one hand, \emph{equality of opportunity} (also known as \emph{formal equality} in law) demands equal treatment for everyone regardless of individual circumstances \citep{sep-equal-opportunity,DBLP:conf/nips/HardtPNS16}.
On the other hand, \emph{equality of outcome} (also known as \emph{substantive equality} in law) acknowledges the existence of initial disparities and seeks to create an equal final outcome for everyone by leveling them \citep{barnard2000substantive,moulin2004fair}. These two principles stand in an inherent conflict with each other, as countering initial disparities oftentimes necessitates the prioritization of those starting from an inferior position, contradicting the principle of equal treatment.\footnote{This conflict recently surfaced in an executive order of US President Trump \citep{trump2025equality}. }

In fair allocation, an overwhelming majority of works adhere to equality of opportunity: They assess fairness by how the allocation distributes resources---irrespective of agents’ starting positions \citep{DBLP:journals/ai/AmanatidisABFLMVW23,DBLP:reference/choice/BouveretCM16,DBLP:journals/jair/LiuLSW24}.
For example, under the classic notion of envy-freeness, an allocation is considered fair if no agent prefers the bundle allocated to another agent to their own.
However, this ignores initial inequalities prevalent in many application areas of fair allocation.
For example, in medical resource allocation programs, patients have different initial probabilities of recovery; educational interventions serve students with varying levels of prior knowledge or preparedness; and more generally, support programs often address individuals from diverse backgrounds and starting positions~\citep{csdh2008,dong2020does,emanuel2023shared}.

Motivated by the prevalence of such disparities, this paper initiates the study of equality of outcome in algorithmic fair allocation. In many real-world allocation problems, decision makers aim for equal post-allocation outcomes instead of equal treatment \citep{LANE201711,scottishgov2021fairer}. One example is the Fairer Scotland Duty \citep{scottishgov2021fairer},
which places a statutory responsibility on certain public bodies to consider how their decisions can help reduce inequalities of outcome.
Notably, existing, intensively studied formal models of fairness reflecting equality of opportunity fail to capture this goal of equality of outcome, a gap that we address in this work. To give another concrete motivating example, consider the following real-world case: the Indian NGO ARMMAN runs large-scale mobile health programs for new mothers, enrolling millions across the country \citep{armman2025}. To boost program engagement, ARMMAN can make a limited number of weekly calls to participants, a sequential resource allocation problem \citep{DBLP:journals/aim/VermaSMVGMHTJTT23}. Crucially, ARMMAN seeks to prioritize calls to mothers with lower initial engagement, often from historically marginalized communities: Rather than allocating calls proportionally across demographics, their goal is to ensure that engagement levels between different demographic groups are equalized post-intervention \citep{DBLP:conf/uai/VermaZSBTT24}, thus following the principle of equality of outcome.

\subsection{Our Contributions}
We initiate the formal study of fair allocation for agents with initial disparities, focusing on fairness notions aligned with the principle of equality of outcome. We do so by equipping each agent $i$ with an agent-specific \emph{initial utility} value $b_i\in \mathbb{R}_{\geq 0}$ and measuring fairness in terms of the utility the agent derives from their bundle plus $b_i$. Throughout, we restrict attention to additive utilities, as is a standard and well-established choice in the literature, particularly when introducing new settings or fairness notions~\citep{DBLP:journals/ai/AmanatidisABFLMVW23,DBLP:journals/sigecom/AzizLMW22,DBLP:conf/aamas/BarmanG0JN19,DBLP:journals/aamas/BouveretL16,DBLP:journals/teco/ChakrabortyISZ21,DBLP:journals/jacm/KurokawaPW18}, as done in our work.

In \Cref{first-attempt}, we present direct adaptations of classical fairness notions -- envy-freeness (EF) and envy-freeness up to one resource (EF1) -- to our setting, which we term EF-init and EF1-init, respectively. Concretely, for EF-init (EF1-init), we require for every pair of agents $i$ and $j$ that agent $i$'s initial utility plus $i$'s utility for their own bundle is at least as large as agent $j$'s initial utility plus agent $i$'s utility for $j$'s bundle (after removing a resource).
Our results reveal that introducing initial utilities fundamentally changes the nature of the problem: even in simple cases with identical resources, complete EF1-init allocations may fail to exist, and deciding whether such an allocation exists is NP-complete. This stands in stark contrast to the classical setting, where complete EF1 allocations are always guaranteed and efficiently computable.

Nevertheless, we show that for a constant number of agents (under unary encoding of utility values), the existence of a complete EF-init and EF1-init allocations can be decided efficiently using dynamic programming. Furthermore, for the special case of identical resources, we present a polynomial-time algorithm for deciding the existence of EF-init allocations.
From a practical standpoint, the allocation of identical resources is a natural and frequently encountered scenario, for example when allocating treatments or interventions (e.g., the allocation task faced by ARMMAN described above).
From a theoretical perspective, our result highlights the added complexity introduced by initial utilities: while in the classical setting without initial utilities any envy-free allocation simply assigns the same number of resources to each agent, in our model, ensuring fairness requires a technically involved dynamic programming approach and a meticulous analysis of the interplay between initial utilities and allocated bundles.

The non-existence of complete EF1-init allocations limits their practical usefulness and highlights the need for fundamentally new approaches. In \Cref{sec:satisfiable}, to address this limitation, we search for an envy-based fairness criterion that is always satisfiable. We present minimum-EF1-init, a new envy-based fairness notion that is guaranteed to exist.
The core idea is to relax EF1-init when evaluating whether an agent $i$  with higher initial utility envies an agent $j$ with lower initial utility.
Instead of comparing the sums of the initial utilities and $i$'s values for the bundles, we take a different approach: We first allow $j$ to level the initial utility difference between $i$ and $j$ with a subset $X^*$ of $j$'s resources which, from $j$'s perspective, equalizes this difference. Then, we compare $i$'s bundle to $j$'s bundle without $X^*$ under EF1.
We show that minimum-EF1-init coincides with EF1-init in settings where resources' usefulness is diminishing in initial utility. 
We develop an adaptation of the classic round-robin algorithm for the setting with initial utilities that computes an allocation satisfying minimum-EF1-init. In our version, agents still select their most preferred available resource in rounds, but participation in each round is restricted: Initially, only agents with the lowest initial utility are present, and additional agents are added to the picking order as soon as all present agents have reached their initial utility.
While our algorithm is simple to state, making it an appealing distribution algorithm in and of itself, the proof that it guarantees  minimum-EF1-init is more involved, requiring a careful analysis of the position where newly added agents are inserted into the picking order.

\subsection{Related Work}\label{sub:rel}

There is a rich body of work studying the fair allocation of indivisible resources, considering various preference models and fairness concepts \citep{DBLP:journals/ai/AmanatidisABFLMVW23,DBLP:reference/choice/BouveretCM16,DBLP:journals/jair/LiuLSW24}. Our work aligns with two major strands in this literature.
First, alongside share-based notions such as the maximin share, envy-based concepts--as considered here--have arguably received the most attention \citep{DBLP:journals/ai/AmanatidisABFLMVW23,DBLP:conf/sigecom/LiptonMMS04,DBLP:conf/ec/CaragiannisGH19}.
Second, additive utility functions constitute the standard preference model in this domain due to their relative simplicity and expressive power \citep{DBLP:journals/teco/CaragiannisKMPS19,DBLP:journals/ai/AmanatidisABFLMVW23,DBLP:reference/choice/BouveretCM16}.

While fair allocation has been studied extensively, there are, to the best of our knowledge, only two works that extend fair allocation models to account for generally differing starting positions prior to the allocation process  \citep{DBLP:conf/aaai/HV0V25,DBLP:conf/aaai/DeligkasEGGI25}.
Both focus on the following problem: given a partial initial allocation, can it be extended to a complete allocation satisfying a specified, traditional fairness criterion?
They analyze the (parameterized) computational complexity of this extension problem under various preference models and fairness notions, highlighting its general intractability.
Our work relates to theirs in that our model can be viewed as a special case of theirs, where all agents assign identical utilities to the initially allocated resources and each agent holds at most one resource in the initial allocation. As a result, some of their positive complexity results for restricted cases carry over to our setting (see Appendix~\ref{app:RW} for a detailed discussion).
However, our study fundamentally differs in two key aspects: First, instead of focusing on computational aspects of existing fairness axioms, we design new fairness notions (following the principle of equality of outcome) tailored to the setting with initial utilities. Second, by putting a focus on the case of initial utilities, we are able to obtain stronger axiomatic guarantees more relevant to our setting. Our negative results can be seen as strengthened versions of analogous results in the setting considered by  \citet{DBLP:conf/aaai/HV0V25} and \citet{DBLP:conf/aaai/DeligkasEGGI25}. In a related direction, \citet{ijcai2025p417} investigate a setting where the set of resources available for completion does not need to be fully allocated and can potentially be unbounded.

Moreover, our work connects to the study of fair allocation with subsidies \citep{DBLP:conf/sagt/HalpernS19}, in which an additional divisible resource can be used to make an allocation of indivisible resources envy-free, and to lines of work on fair allocation with agents with different entitlements  \citep{DBLP:journals/jair/FarhadiGHLPSSY19,DBLP:journals/teco/ChakrabortyISZ21}  and budget constraints \citep{DBLP:journals/corr/abs-2012-03766,DBLP:conf/aaai/Barman0SS23}.
More broadly, our work fits into the study of completion problems in computational social choice, such as possible and necessary winner problems in voting \citep{DBLP:journals/jair/XiaC11} and stable matching \citep{DBLP:journals/tcs/DiasFFS03,DBLP:journals/tcs/0001BFGHMR22}.

\section{Preliminaries}
Let $[\ell]\coloneqq\{1, \dots, \ell\}$ for $\ell \in \mathbb{N}$. We define $\RRR \coloneqq \{x \in \RR \mid x\geq 0\}$. To highlight that an inequality follows from a previous result, definition, or equation, we add a reference in the subscript, e.g., by writing $\leq_{(*)}$ for an inequality that follows from a previously introduced equation $(*)$.

We consider the problem of allocating a set $R$ of $m$ \emph{indivisible resources} among a set $A \coloneqq [n]$ of $n \geq 2$ \emph{agents}.  We refer to subsets $X\subseteq R$ of resources as \emph{bundles} and
denote by $2^R$ the set of all bundles.
Each agent $i\in A$ has a \emph{utility function} $u_i : 2^R \to \RRR$, where $u_i(\emptyset)=0$.\footnote{To rule out trivial edge cases, we assume that for every agent $i \in A$, there exists some resource $r \in R$ such that $u_i(\{r\})>0$, and for every resource $r \in R$, there exists some agent $i \in A$ such that $u_i(\{r\}) > 0$.}\textsuperscript{,}\footnote{Regarding the encoding of utility functions, all our hardness results hold even if utility functions are encoded in unary. If not stated otherwise, our positive results hold for both unary- and binary-encoded utility functions.}
We assume throughout that utility functions are \emph{additive}, i.e., $u_i(X)= \sum_{r \in X} u_i(\{r\})$ for any $i\in A$ and $X \subseteq R$.
In certain cases, we consider \emph{identical resources}, where
$u_i(\{r\})=u_i(\{r'\})$ for all $i \in A$ and resources $r,r' \in R$. An \emph{allocation} $\mathcal{X}$ is a tuple of $n$ disjoint bundles $(X_1, \dots, X_n)$ such that $X_i \cap X_j = \emptyset$ for all distinct $i,j \in A$, where $X_i$ is \emph{allocated/assigned} to agent $i\in A$.
An allocation $\mathcal{X}$ is \emph{complete} if $\bigcup_{i\in A} X_i = R$, which can be viewed as a very moderate efficiency criterion.
We study the problem of finding fair and complete allocations.
Two popular fairness notions for indivisible resources are \emph{envy-freeness} (EF) and its relaxation, \emph{envy-freeness up to one resource} (EF1). An allocation $\mathcal{X}$ is EF if $u_i(X_i)\geq u_i(X_j)$ for each $i,j\in A$. It is EF1 if, for every pair of agents $i,j \in A$, there is a resource $r\in X_j$ such that $ u_i(X_i) \geq  u_i(X_j\setminus \{r\})$ or $X_j=\emptyset$. For a variant of envy-freeness, we say that an agent $i\in A$ envies another agent $j\in A$ \emph{under} this variant, if the pair $i,j \in A$ violates the conditions of the variant, e.g., $i$ envies $j$ under EF if $u_i(X_i)<u_i(X_j)$.

\paragraph{Initial Utilities}
We study the new problem of fairly allocating indivisible resources in the presence of \emph{initial utilities}. That is, we assume that every agent $i\in A$ has an initial utility $b_i\in \RRR$.
Thus, after the allocation $\mathcal{X}$, agent $i$ attains utility $b_i + u_i(X_i)$.
Conceptually, we assume that initial utilities are known to the allocation algorithm and comparable across agents (we address this assumption in more detail in the discussion section), i.e., unlike with bundles to which two agents might assign different utilities, the agents agree on their initial utility values. The classical setting is recovered by taking $b_i = 0$ for all $i \in A$.
We group agents by initial utility as follows:
\begin{definition}\label{def:levels}
  The agents are partitioned into $t\in [n]$ \emph{levels} $L_1, \dots, L_t$, where each level $L_h$ with $h \in [t]$  contains all agents with the same initial utility. The levels are indexed such that $b_i < b_j$ for any $i \in L_h$, $j \in L_{h'}$, and $h < h'$.
\end{definition}

  \section{A First Attempt at Adapting Envy Notions to Initial Utilities} \label{first-attempt} 

  We present first adaptations of the classical fairness notions of \emph{envy-freeness}  and its relaxation, \emph{envy-freeness up to one resource}, to the setting with initial utilities. Following the equality of outcome principle, we are interested in measuring the fairness of an allocation $\mathcal{X}$ in terms of the total utility of agents after resources have been allocated, i.e., their initial utility plus the utility they derive from $\mathcal{X}$. Following this rationale, we adapt EF and EF1 as follows:
  \begin{definition}[EF-init]\label{def:EFinit}
    An allocation $\mathcal{X}$ is \emph{EF-init}, if for every pair of agents $i,j \in A$ either $X_j = \emptyset$\footnote{This condition implies that an agent $i\in A$ with $b_i+u_i(X_i)<b_j$ does not envy an agent $j$ under EF-init or EF1-init if $X_j = \emptyset$, as otherwise achieving an envy-free allocation would be impossible in settings with large initial utility disparities.} or it holds that
    $b_i + u_i(X_i) \geq b_j + u_i(X_j)$.
  \end{definition}
  \begin{definition}[EF1-init]\label{def:EF1init}
    An allocation $\mathcal{X}$ is \emph{EF1-init}, if for every pair of agents $i,j \in A$ either $X_j = \emptyset$ or there exists a resource $r\in X_j$ such that
    $
    b_i + u_i(X_i) \geq b_j + u_i(X_j\setminus \{r\})
    $.
  \end{definition}
  In the classical setting without initial utilities, a complete EF1 allocation always exists, while a complete EF allocation may not exist, even with two agents and a single resource. In contrast, with initial utilities, even complete EF1-init allocations fail to exist:
  \begin{observation}\label{prop:EF1_init:nonex}
    There exists an instance with two agents and identical resources in which no complete allocation satisfies EF1-init.
  \end{observation}
  \begin{proof}
    Consider four identical resources $R$ and two agents: agent $j$ with $b_j = 1$ and $u_j(\{r\}) = 3$ for all $r\in R$, and agent $i$ with $b_i = 10$ and $u_i(\{r\}) = 10$ for all $r \in R$. In any complete allocation where $i$ receives at least one resource, $j$ has to get at least three resources, as otherwise $j$ envies $i$ under EF1-init. However, in allocation $\mathcal{X}$ where $j$ receives three and $i$ receives  one resource, $i$ envies $j$  under EF1-init, as $
    b_i + u_i(X_i) = 10+10<1 + 2\cdot 10 = b_j +u_i(X_j\setminus \{r\})$
    for any $r \in X_j$.
  \end{proof}

  Given the non-existence of complete EF-init and EF1-init allocations, we analyze the complexity of determining whether a given instance with initial utilities admits such an allocation. We refer to these computational problems as \EFEX{}  and \EFOEX{}, respectively.

  \paragraph{\EFOEX{}}
  We begin with EF1-init allocations, noting that, in the classical setting, EF1 allocations can be found via a simple round-robin algorithm in polynomial time. However, with initial utilities, the problem becomes computationally hard:

  \begin{restatable}{theorem}{EFEXHard}\label{prop:EFEX_hard}
    \EFOEX{} is NP-complete.
  \end{restatable}
  \newcommand{\EQC}{\textsc{Equitable Coloring}}
    Our reduction is from \EQC{}. To define \EQC{}, we first introduce some necessary definitions.
    \begin{definition}
      For some $\ell \in \NN$, an \emph{$\ell$-coloring} of a graph $G=(V,E)$ assigns a color $c(v)\in [\ell]$ to every vertex $v\in V$. We call an $\ell$-coloring \emph{proper} if it holds for every edge $e=\{u,v\}$ that $c(u)\neq c(v)$. The \emph{color class} $V(h)\subseteq V$ of color $h \in [\ell]$ is the subset of vertices colored with color $h$, formally, $V(h) \coloneqq \{ v \in V \mid c(v) = h \}$.
      We say that an $\ell$-coloring is \emph{equitable} if it is proper and for any two colors $h,h' \in [\ell]$, it holds that $|V(h)| \leq |V(h')|+1$.
    \end{definition}
    For some given $\ell \in \NN$, the \EQC{} problem asks whether a given graph admits an equitable $\ell$-coloring.
    \problemdef{\EQC{}}
    {A graph $G=(V,E)$ and a number $\ell \in \NN$. }
    {Is there an equitable $\ell$-coloring of $G$?}
    \newcommand{\GC}{\textsc{Graph Coloring}}

    The NP-hardness of \EQC{} can be shown by a straightforward reduction (outlined below) from \GC{}, which is defined as follows and is known to be NP-hard (see, e.g., \citet{gareyJ79}).
    \problemdef{\GC{}}
    {A graph $G=(V,E)$ and a number $\ell \in \NN$. }
    {Is there a proper $\ell$-coloring of $G$?}
    We shortly sketch the reduction from \GC{} to \EQC{}. Given a graph $G=(V,E)$ and $\ell \in \NN$, we construct a graph $G'$ from $G$ by adding $(\ell-1) \cdot |V|$ isolated vertices. Clearly, $G$ admits a proper $\ell$-coloring if and only if $G'$ admits an equitable $\ell$-coloring:
    Any proper $\ell$-coloring of $G$ can be extended to an equitable $\ell$-coloring of $G'$ by coloring the added isolated vertices such that every color class has size $|V|$.
    Conversely, any equitable $\ell$-coloring of $G'$ needs to induce a proper  $\ell$-coloring when restricted to $G$.

    Now, we are ready to prove that \EFOEX{} is NP-complete. Our reduction is inspired by a reduction due to \citet{hosseiniSVWX20}\footnote{The reduction can be found in the full version of their paper \citep{DBLP:journals/corr/abs-1907-02583}.}, who reduce from a variant of \EQC{} where all color classes are required to be of the same size to prove NP-hardness of deciding the existence of an envy-free allocation (without initial utilities) for binary utilities\footnote{Agents either approve or disapprove a resource. The utility for a bundle $X\subseteq R$ is the number of approved resources in $X$.}.
    \begin{proof}[Proof of \Cref{prop:EFEX_hard}]
      First, we observe that \EFOEX{} is in NP: Given an allocation $\mathcal{X}$, we can check for any pair of agents $i,j \in A$ whether agent $i$ envies $j$ under EF1-init.
      To show NP-hardness, we reduce from \EQC{}. Given a graph $G=(V,E)$ and a number $\ell \in \NN$, we construct an instance for \EFOEX{} with agents $A$ and resources $R$ as follows:
      \begin{itemize}
        \item The set of agents $A$ is the union of two sets $A_E$ and $A_C$. The set $A_E$ contains an \emph{edge agent} $e$ for each edge $e\in E$. $A_C$ contains a \emph{color agent} $h$ for every color $h \in [\ell]$.
        \item For every vertex $v\in V$, the set $R$ contains a \emph{vertex resource} $v$.
        \item The utility functions and the initial utilities of the agents are defined as follows. All color agents from $A_C$ have an initial utility of $0$ and get a utility of $1$ for every resource, that is, $u_i(X)=|X|$ for any $i\in A_C$ and $X\subseteq R$. All edge agents $e \in A_E$ have an initial utility of $|V|+1$.
          An edge agent $e$ for edge $e=\{u,v\}\in E$ gets a utility of $|V|+2$ for the two vertex resources corresponding to edge $e$'s endpoints $u$ and $v$. All other resources do not increase the utility of $e$. Formally, $u_e(X)=|e\cap X|\cdot (|V|+2)$ for any $X\subseteq R$.
      \end{itemize}
      This construction can clearly be computed in polynomial-time.
      It remains to prove that $G$ admits an equitable $\ell$-coloring if and only if the constructed instance has a complete EF1-init allocation.
      \paragraph{($\Rightarrow$):}
      Given an equitable $\ell$-coloring of $G$, we construct an allocation $\mathcal{X}$ by assigning the vertex resource $v$ to the color agent $c(v)$.
      As every vertex is colored, the constructed allocation $\mathcal{X}$ is complete. Furthermore, note that only color agents receive resources.

      Next, we show that $\mathcal{X}$ is EF1-init.  Consider any two  agents $i,j \in A$.
      If $j \in A_E$, then, as $X_j=\emptyset$, agent $i$ does not envy agent $j$ under EF1-init.
      Next, suppose that $i \in A_E$ and $j \in A_C$. If $X_j=\emptyset$, then agent $i$ does not envy agent $j$ under EF1-init. Otherwise, we choose a resource $r\in X_j$ as follows. First, assume that $X_j \cap i \neq \emptyset$, that is, one of the endpoints of the edge corresponding to $i$ is assigned color $j$. Then, we choose a resource $r\in  X_j \cap i$, otherwise, we choose an arbitrary resource $r\in X_j$. Note that it needs to hold that $|X_j \cap i|\leq 1$, as the coloring would otherwise not be proper. Thus, we have that $(X_j\setminus \{r\}) \cap i = \emptyset$ and thus $
      b_j + u_i(X_j\setminus \{r\})=0.
      $
      Therefore, agent $i$ does not envy agent $j$ under EF1-init.

      Finally, assume that $i,j \in A_C$. Since the coloring is equitable, we have that $|X_j|\leq |X_i|+1$.
      Therefore, either $X_j=\emptyset$ or it holds for any $r\in X_j$ that
      \[
        b_i+ u_i(X_i)=u_i(X_i)=|X_i|\geq |X_j|-1 \geq b_j + u_i(X_j\setminus \{r\}),
      \]
      so agent $i$ does not envy agent $j$ under EF1-init.
      As we have exhausted all possible cases, it follows that $\mathcal{X}$ is EF1-init.
      \paragraph{($\Leftarrow$):}
      Let $\mathcal{X}$ be any complete EF1-init allocation for the constructed instance.
      We first prove that no edge agent can have a resource in $\mathcal{X}$. Suppose for the sake of contradiction that there is an edge agent $j \in A_E$ with $X_j\neq \emptyset$. Pick any color agent $i \in A_C$. Note that for all $r\in X_j$, it holds that
      \[
        b_i+ u_i(X_i)<|V|<|V|+1 = b_j \leq  b_j + u_i(X_j \setminus \{r\}).
      \]
      Thus, agent $i$ envies agent $j$ under EF1-init. Therefore, only color agents can be allocated resources in $\mathcal{X}$.

      Next, we construct an $\ell$-coloring in $G$ from the allocation $\mathcal{X}$ by coloring every vertex $v\in V$ with the color corresponding to the color agent $h \in A_C$ such that $v\in X_h$. As no edge agent can have a resource and $\mathcal{X}$ is complete, every vertex is assigned a color.
      First, we argue that the coloring is proper. Consider any edge $e=\{u,v\}\in E$. Suppose that $u$ and $v$ have the same color $h$. Then, $e \subseteq X_h$. However, this implies that the edge agent for $e$ envies the color agent $h$ under EF1-init, as $X_h\neq \emptyset$, and for all $r\in X_h$,
      \[
        b_e + u_e(X_e) = |V| + 1 < |V|+2 \leq b_h + u_e(X_h \setminus \{r\}).
      \]
      This contradicts that $\mathcal{X}$ is EF1-init, therefore the constructed coloring is proper.

      It remains to show that the coloring is also equitable. Assume for the sake of contradiction that there are two colors $h,k \in [\ell]$ such that $|V(h)|>|V(k)|+1$. However, this implies that the color agent $k$ envies color agent $h$ under EF1-init. We have that  $X_h\neq \emptyset$, and for all $r\in X_h$, it holds that
      \[
        b_k + u_k(X_k) = |X_k| < |X_h|-1 = b_h + u_k(X_h\setminus \{r\}).
      \]
      Thus, the constructed coloring is an equitable $\ell$-coloring of $G$, which completes the proof.
    \end{proof}
  One can check that our hardness proof also applies to the special case where, for each $i\in A$, we have $b_i\in \{ 0, x\}$ for some $x \in \RRR$.

  \paragraph{\EFOEX{}/\EFEX{} for Few Agents}
  On the positive side, following a dynamic programming approach, deciding the existence of a complete EF1-init or EF-init allocation becomes tractable when the number of agents is constant:
  \begin{restatable}{proposition}{EFEXPoly}
    \label{prop:EFEX_poly}
    For a unary encoding of the utility functions and a constant number of agents, \EFOEX{} and \EFEX{} are polynomial-time solvable.
  \end{restatable}
  \begin{proof}
      Both problems can be solved by a similar dynamic programming approach for a constant number of agents $n=|A|$. We first prove the statement for EF-init, and then show how the proof can be adapted for EF1-init. We fix  an arbitrary ordering of the resources, i.e., $R=\{r_1, r_2, \dots, r_m\}$. Since we assume a unary encoding of the utility functions, without loss of generality, we assume that $u_i(X)\in \NN$ for all $i \in A$ and $X\subseteq R$, and define $s\coloneqq\max_{i\in A} u_i(R)$.

      The dynamic program for EF-init is based on
      a boolean  table $T[\ell,(v_{i,j})_{i,j=1}^n]$, where $\ell \in [m]$ and $v_{i,j} \in \{0,\dots,s\}$ for all $i,j \in A$.
      An entry should be set to "true" if there is an allocation $\mathcal{X}$ of the resources $\{r_1, \dots, r_\ell\}$, where it holds that $u_i(X_j)=v_{i,j}$ for all $i,j\in A$.
      Note that as $n$ is constant and $s$ is polynomial in the input size since we assume a unary encoding, it follows that the size of the table in $\bigO(m\cdot s^{n^2})$ is polynomial in the size of the input.
      We initialize the entry for $\ell=0$ and $v_{i,j}=0$ for all $i,j\in A$ to "true" and all other entries to "false".
      The table entries are filled in ascending order for $\ell\in \{1, \dots, m\}$ as follows. To fill an entry for given indices $\ell \in  \{1, \dots, m\}$ and $v_{i,j} \in [s]$ with $i,j \in A$, check if there exists an agent $i^* \in A$, such that the entry for $\ell'=\ell-1$, $v'_{j,i^*}=v_{j,i^*}-u_j(\{r_\ell\})$ for all $j\in A$, and $v'_{i,j}=v_{i,j}$ for all $i\in A$ and $j\neq i^*$ is set to "true". Intuitively, this means that the $\ell$-th resource is given to agent $i^*$, and it remains to check if there is an (already computed) entry that represents that the first $\ell-1$ resources can be distributed accordingly. Clearly, each entry can be computed in polynomial time.
      After computing the table, to decide \EFEX{}, check if there is an entry for $\ell=m$, where for all $i,j\in A$, it holds that $b_i + v_{i,i}\geq b_j + v_{i,j}$ or $v_{i', j}=0$ for all $i' \in A$ (note that this implies that $X_j=\emptyset$ in the corresponding allocation $\mathcal{X}$).

      To adapt the approach for EF1-init, we introduce additional dimensions $p_{i,j}$ for each pair of agents $i,j\in A$ used to store the value of the resource in $j$'s bundle that $i$ likes the most, yielding a table $T[\ell,(v_{i,j})_{i,j=1}^n, (p_{i,j})_{i,j=1}^n]$, where $\ell \in [m]$ and $v_{i,j}, p_{i,j} \in \{0,\dots,s\}$ for all $i,j \in A$.
      We initialize the entry for $\ell=0$ and $v_{i,j}=p_{i,j}=0$ for all $i,j\in A$ to "true" and all other entries to "false".
      Again, table entries are filled in ascending order for $\ell\in \{1, \dots, m\}$. Analogously as before, to fill an entry for given indices $\ell \in  \{1, \dots, m\}$ and $v_{i,j},p_{i,j}  \in [s]$ with $i,j \in A$, check that there is an agent $i^* \in A$ such that it holds that $p_{j,i^*}\geq u_j(\{r_\ell\})$ for all $j\in A$, and such that there is a "true"-entry for $i^*$ with the following indices:
      \begin{itemize}
        \item  $\ell'=\ell-1$,
        \item $v'_{j,i^*}=v_{j,i^*}-u_j(\{r_\ell\})$ for all $j\in A$, and $v'_{i,j}=v_{i,j}$ for all $i\in A$ and $j\neq i^*$,
        \item  $p'_{j,i^*}=p_{j,i^*}$ if $p_{j,i^*}> u_j(\{r_\ell\})$ and otherwise $p'_{j,i^*}\in [p_{j,i^*}]$ for all $j\in A$,
        \item $p'_{i,j}=p_{i,j}$ for all $i\in A$ and $j\neq i^*$.
      \end{itemize}
      Again, it is clear that the table size and the time to compute each entry is polynomial in the input size.
      To decide \EFOEX{}, after filling the table, check if there is an entry for $\ell=m$, where for all $i,j\in A$, it holds that $b_i + v_{i,i}\geq b_j + v_{i,j} - p_{i,j}$ or $v_{i', j}=0$ for all $i' \in A$.
    \end{proof}

  \paragraph{\EFEX{}}
  The complexity of deciding the existence of a (complete) envy-free allocation (without initial utilities) is well understood, and the problem is known to be (at least) NP-hard, even for restrictive settings such as if there are only two agents with identical preferences or agents have 0/1 utilities for each resource  \citep{DBLP:journals/jair/BouveretL08}. By setting $b_i = 0$ for all $i \in A$, these results directly carry over to \EFEX{}. We therefore focus on the special case where all resources are identical, obtaining the following result:

  \begin{restatable}{theorem}{EFInitIdentical}
    \label{thm:EFInit_identical}
    \EFEX{} for identical resources can be decided in  $\mathcal{O}(n^2 \cdot m^3)$ time and $\mathcal{O}(n \cdot m^2)$ space.
  \end{restatable}
  We first sketch the idea of the proof.
  To solve this problem, we examine the partitioning of agents by their initial utility into $t$ levels, as introduced in \Cref{def:levels}. Note that in any EF-init allocation, all agents in a level need to get the same number of resources, as all resources are identical. The following observation is crucial to prove the above result. If there are two agents $i,j\in A$ with $b_i<b_j$, where agent $i$ has a strictly smaller utility for each resource than $j$, then agent $j$ cannot get any resources in an EF-init allocation: Otherwise, agent $i$ has to receive some number $x$ of resources more than agent $j$ to equalize the difference $b_j -b_i$ in initial utilities. However, as $j$ has a strictly higher utility for each of these resources, this implies that $j$ envies $i$ under EF-init, as the utility that $j$ assigns to the $x$ additional resources is strictly greater than $b_j - b_i$.
  We call such a pair of agents $i,j \in A$ a \emph{violating pair}.

  Importantly, this observation provides the following useful way to structure the levels. Let $h^*\in [t]$ be the minimum level such that there is an agent $j \in L_{h^*}$  that is part of a violating pair ($j$ is the agent with higher initial utility). If there is no violating pair in the whole instance, set $h^*:=t+1$.
  It follows that all agents in levels $h^*$ and above cannot get any resources.
  Moreover, using the structure implied by the fact that there is no violating pair in the first $t^*\coloneqq h^*-1$ levels, we can show that within these levels, envy-freeness between pairs of agents behaves ``transitively'' with respect to the ordering of the levels.
  \begin{restatable}{lemma}{EFInitTransitive}\label{lemma:EFInit_transitive}
    Let $h,h',h'' \in [t^*]$ with $h'<h<h''$ and consider agents $i^* \in L_h$, $i \in L_{h'}$, and $j \in L_{h''}$.
    If in some allocation $\mathcal{X}$ there is no envy between agents $i$ and $i^*$, and between $i^*$ and $j$ under EF-init, then there is no envy between $i$ and $j$ under EF-init.
  \end{restatable}
  This enables a dynamic programming approach to decide \EFEX{} for identical resources. 

  In the following, we provide the full proof of \Cref{thm:EFInit_identical}.
  For agent $i\in A$,  let $v_i \coloneqq u_i(\{r\})$ for any $r\in R$ be the \emph{value} that agent $i$ derives from a single resource. We begin by formalizing the following observations, which we informally remarked above.
    \begin{observation}\label{EFinit:ordered}
      Let $h\in [t]$. Let $\mathcal{X}$ be an allocation where there is no envy between any pair of agents $i,j \in \bigcup_{\ell \in [h]} L_\ell$ in the first $h$ levels.  Then, the following holds.
      \begin{enumerate}
        \item If for some $h' \in [h]$, there is an agent $i \in L_{h'}$ with $|X_i|\geq 1$, then $|X_j|\geq 1$ for all $j\in L_{h''}$ with $h''\leq h'$.
        \item For any level $L_{h'}$ with $h'\in [h]$, it holds that any two agents $i,j\in L_{h'}$ are allocated the same number of resources $|X_i|=|X_j|$.
        \item Let $i,j \in \bigcup_{\ell \in [h]} L_\ell$ be two agents in the first $h$ levels. If $b_i<b_j$ and $v_i<v_j$, we have that $|X_j|=0$.
      \end{enumerate}
    \end{observation}
    \begin{proof}
      It is easy to see that the first two properties need to hold. To see why the third property holds, consider two agents $i,j \in \bigcup_{\ell \in [h]} L_\ell$ with $b_i<b_j$ and $v_i<v_j$. Now, assume for the sake of contradiction that $|X_j|>0$. We have that $|X_i|>0$, since otherwise agent $i$  envies agent $j$ under EF-init, as $i$ has a lower initial utility. Since agent $i$ does not envy $j$ under EF-init in $\mathcal{X}$, we have that
      \begin{align*}
        & b_i + u_i(X_i) \geq b_j + u_i(X_j)\\
        \iff& b_i + v_i\cdot |X_i| \geq b_j + v_i\cdot |X_j|\\
        \iff& |X_i| -|X_j| \geq \frac{b_j -b_i}{v_i}.
      \end{align*}
      However, since agent $j$ does not envy $i$ under EF-init, we need to have that
      \begin{align*}
        & b_j + u_j(X_j) \geq b_i + u_j(X_i)\\
        \iff& b_j + v_j\cdot |X_j| \geq b_i + v_j\cdot |X_i|\\
        \iff& \frac{b_j-b_i}{v_j}\geq |X_i|-|X_j|.
      \end{align*}
      Note that this is a contradiction: Since $b_j-b_i>0$ and $v_i<v_j$, we have that
      \[
        |X_i|-|X_j|\leq\frac{b_j-b_i}{v_j}<\frac{b_j-b_i}{v_i}\leq |X_i|-|X_j|.
      \]
      Therefore, it needs to hold that $|X_j|=0$.
    \end{proof}
    Note that for $h=t$, the allocation $\mathcal{X}$ from \Cref{EFinit:ordered} is EF-init. Then, the properties hold for any level and all agents.

    Let $h^*\in [t]$ be the minimal level such that there is a violating pair of agents $i,j\in A$ with $b_i<b_j$ and $v_i<v_j$ where $j\in L_{h^*}$ (if such a pair exists). Then, by  \Cref{EFinit:ordered}, it needs to hold for any EF-init allocation $\mathcal{X}$ and agent $i' \in L_{h'}$ with $h'\geq h^*$ that $|X_{i'}|=0$. That is, no agent outside the first $h^*-1$ levels gets a resource in any EF-init allocation. If $h^*=1$, we know that the given instance for \EFEX{} is a no-instance, since there cannot be a complete EF-init allocation.

    Otherwise, we define $t^*\coloneqq h^*-1$. Let $i\in \bigcup_{\ell \in \{h^*, \dots, t\}} L_{\ell}$ be an agent outside the first $t^*$ levels. For every level $L_{h'}$ with $h'\in [t^*]$, we know that an agent $j \in L_{h'}$ can get at most $\lfloor \frac{b_i - b_j}{v_i} \rfloor$ resources, as otherwise agent $i$ would envy $j$ under EF-init. By defining $k_{h'}$ as the minimum of these values for all $i\in \bigcup_{\ell \in \{h^*, \dots, t\}} L_{\ell}$, we get an upper-bound on the number of allocated resources for each agent in the first $t^*$ levels in any EF-init allocation.
    If there is no violating pair in the whole instance, we set $t^*\coloneqq t$ and $k_{h'}\coloneqq m$ for all $h' \in [t]$.

    Now, we can show the crucial lemma that we already stated above.
    \EFInitTransitive*
    \begin{proof}
      We prove the statement by proving the following two implications.
      \begin{enumerate}
        \item If $i$ does not envy $i^*$ and $i^*$ does not envy $j$, then $i$ does not envy $j$ (under EF-init).
        \item If $j$ does not envy $i^*$ and $i^*$ does not envy $i$, then $j$ does not envy $i$ (under EF-init).
      \end{enumerate}
      Firstly, assume that  $i$ does not envy $i^*$ and $i^*$ does not envy $j$ under EF-init.
      If $X_{i^*}=\emptyset$, then also $X_{j}=\emptyset$, as $i^*$ does not envy $j$ under EF-init and $b_j>b_{i^*}$, which implies that $i$ does not envy $j$ under EF-init. Thus, we assume that $X_{i^*}\neq \emptyset$ and $X_{j}\neq\emptyset$.
      As $i$ does not envy $i^*$ under EF-init, we have that
      \begin{align*}
        b_{i} + u_{i}(X_{i}) \geq b_{i^*} + u_{i}(X_{i^*}) &\iff b_{i} + v_{i} \cdot |X_{i}| \geq b_{i^*} + v_{i}\cdot |X_{i^*}|\\
        &\iff   |X_{i}| -|X_{i^*}| \geq \frac{b_{i^*} - b_{i}}{v_{i}}.
      \end{align*}

      Furthermore, agent $i^*$ does not envy $j$ under EF-init, so
      \begin{align*}
        &b_{i^*} + u_{i^*}(X_{i^*}) \geq b_{j} + u_{i^*}(X_{j})\\ \iff &b_{i^*} + v_{i^*} \cdot |X_{i^*}| \geq b_{j} + v_{i^*}\cdot |X_{j}|\\
        \Rightarrow   &|X_{i^*}| -|X_{j}| \geq \frac{b_{j} - b_{i^*}}{v_{i^*}} \geq \frac{b_{j} - b_{i^*}}{v_{i}},
      \end{align*}
      where the last inequality holds since $b_{j} - b_{i^*}>0$ and $v_{i^*} \leq v_{i}$ (as there is no violating pair in the first $t^*$ levels and $h'<h$).
      This implies that $i$ does not envy $j$ under EF-init, since
      \begin{align*}
        |X_{i}| -|X_{j}| = |X_{i}| -|X_{i^*}| + |X_{i^*}| -|X_{j}|\\\geq \frac{b_{j} - b_{i^*}}{v_{i}} + \frac{b_{i^*} - b_{i}}{v_{i}} = \frac{b_{j} - b_{i}}{v_{i}}\\
        \iff b_i + v_i \cdot |X_{i}| \geq b_j + v_i \cdot |X_{j}|.
      \end{align*}
      Secondly, assume symmetrically that $j$ does not envy $i^*$ and $i^*$ does not envy $i$ under EF-init.
      Since $j$ does not envy $i^*$ under EF-init, if $X_{i^*}\neq\emptyset$, then it holds that
      \begin{align*}
        b_{j} + u_{j}(X_{j}) \geq b_{i^*} + u_{j}(X_{i^*}) &\iff b_{j} + v_{j} \cdot |X_{j}| \geq b_{i^*} + v_{j}\cdot |X_{i^*}|\\
        &\iff   \frac{b_{j} - b_{i^*}}{v_{j}}\geq |X_{i^*}| -|X_{j}|.
      \end{align*}
      If $X_{i^*}=\emptyset$, then $|X_{i^*}| -|X_{j}| \leq 0 <  \frac{b_{j} - b_{i^*}}{v_{j}}$, so in any case, it holds that $\frac{b_{j} - b_{i^*}}{v_{j}}\geq |X_{i^*}| -|X_{j}|$.

      Furthermore, we have that $i^*$ does not envy $i$ under EF-init. If $X_i=\emptyset$, then this directly implies that $j$ does not envy $i$ under EF-init. Otherwise, it follows that
      \begin{align*}
        b_{i^*} + u_{i^*}(X_{i^*}) \geq b_{i} + u_{i^*}(X_{i}) &\iff b_{i^*} + v_{i^*} \cdot |X_{i^*}| \geq b_{i} + v_{i^*}\cdot |X_{i}|\\
        &\Rightarrow  |X_{i}| -|X_{i^*}| \leq \frac{b_{i^*} - b_{i}}{v_{i^*}} \leq \frac{b_{i^*} - b_{i}}{v_{j}},
      \end{align*}
      where the last inequality holds since $b_{i^*} - b_{i}>0$ and $v_{i^*} \geq v_{j}$ (since there is no violating pair in the first $t^*$ levels and $h''>h$).
      This implies that $j$ does not envy $i$ under EF-init, since
      \begin{align*}
        |X_{i}| -|X_{j}| = |X_{i}| -|X_{i^*}| + |X_{i^*}| -|X_{j}|\\\leq \frac{b_{i^*} - b_{i}}{v_{j}} + \frac{b_{j} - b_{i^*}}{v_{j}} = \frac{b_{j} - b_{i}}{v_{j}}\\
        \iff b_i + v_j \cdot |X_i| \leq b_j + v_j \cdot |X_j|.
      \end{align*}
    \end{proof}
    Now, we are ready to prove that \EFEX{} for identical resources and additive utility functions can be decided in time polynomial in $n$ and $m$.
    \begin{proof}[Proof of \Cref{thm:EFInit_identical}]
      We prove the statement by giving a dynamic-programming algorithm based on the following table $D$ of size $\mathcal{O}(n\cdot m^2)$. For $a\in [m]$, $b \in [t^*]$, and $c \in [m]$, we store in entry $D[a][b][c]$ whether there is an allocation $\mathcal{X}$ that satisfies the following requirements, which we call \emph{fitting} for this entry. The table entry may be omitted when it is clear from context.
      \begin{itemize}
        \item Allocation $\mathcal{X}$ allocates exactly $a$ resources,
        \item  the set of agents that get at least one resource is the union of the first $b$ levels $\bigcup_{\ell \in [b]} L_\ell$,
        \item every agent $i\in L_b$ gets exactly $c=|X_i|$ resources,
        \item there is no pair of agents $i,j\in\bigcup_{\ell \in [b]} L_\ell$ in the first $b$ levels such that $i$ envies $j$ in $\mathcal{X}$ under EF-init,
        \item and it holds that $|X_i|\leq k_{h}$ for any $h \in [t^*]$ and agent $i \in L_h$ (recall that we defined these upper-bounds above when defining $t^*$).
      \end{itemize}
      Next, we give \Cref{alg:efInit} to decide \EFEX{}. The algorithm first fills table $D$ and then checks if there is a "yes"-entry for which a fitting allocation is EF-init. Note that for a fitting allocation for an entry for level $b\in [t^*]$, we only require that there is no envy under EF-init between agents  $i,j\in\bigcup_{\ell \in [b]} L_\ell$ in the first $b$ levels, so we have to check if agents in higher levels are envious.
      \begin{algorithm}[t!]
        \caption{Decide \EFEX{}}\label{alg:efInit}
        \begin{algorithmic}[1]
          \State \textbf{Global Table:} $D[a][b][c]$, for $a \in [m]$, $b \in [t^*]$, and $c \in [m]$.

          \For{$a \in [m]$, $b \in [t^*]$, and $c \in [m]$} \label{alg:efInit:initA}
          \State $D[a][b][c] \gets \text{"no"}$ \Comment{Initialize Table $D$ entries to "no".}
          \EndFor

          \For{$a \in [m]$ and $c \in [k_1]$} \label{alg:efInit:initB}
          \If{$|L_1| \cdot c = a$}
          \State $D[a][1][c] \gets \text{"yes"}$ \Comment{Initialize  first level.} \label{alg:efInit:base}
          \EndIf
          \EndFor

          \For{$b \in \{2, \dots, t^*\}$} \label{alg:efInit:loop}
          \For{$a \in [m]$ and $c \in [k_b]$}
          \If{$a - |L_b| \cdot c \geq 0$}
          \State $a' \gets a - |L_b| \cdot c$
          \State $b' \gets b - 1$
          \LongState{ Choose $i \in L_{b'}$ with minimal $v_i$ and $j \in L_b$ with maximal $v_j$}
          \State $c'_{min} \gets c + \lceil \frac{b_j - b_i}{v_i} \rceil$ \Comment{Ensures $i$ does not envy $j$.}
          \State $c'_{max} \gets c + \lfloor \frac{b_j - b_i}{v_j} \rfloor$ \Comment{Ensures $j$ does not envy $i$.}

          \If{$c'_{min} \leq c'_{max}$}
          \For{$c' \in \{c'_{min}, \dots, c'_{max}\}$}
          \If{$D[a'][b'][c'] = \text{"yes"}$}
          \State $D[a][b][c] \gets \text{"yes"}$ \label{alg:efInit:combine}

          \EndIf
          \EndFor
          \EndIf
          \EndIf
          \EndFor
          \EndFor
          \For{$b\in [t^*]$ and $c \in [k_b]$}\label{alg:efInit:check}
          \If{$D[m][b][c]=\text{"yes"}$ and CheckEFInit($m, b, c$)="yes"}
          \State \textbf{Accept} \label{alg:efInit:accept} \Comment{Accept if fitting allocation is EF-init.}
          \EndIf
          \EndFor
          \State \textbf{Reject}
          \vspace{0.5em}
          \Function{CheckEFInit}{$a, b, c$}
          \State Choose arbitrary $j\in L_b$.
          \If{$b<t^*$}
          \Statex \Comment{%
                \parbox[t]{0.855\linewidth}{\hangindent=0.25em
                     Check if an agent in levels $L_{b+1},\dots, L_{t^*}$ (that receives no resources) envies $j$.}%
        }

          \For{$i\in \bigcup_{h \in \{b+1, \dots, t^*\}} L_h$ }
          \If{$c > \frac{b_{j}-b_i}{v_i}$}
          \Return "no"
          \EndIf
          \EndFor
          \EndIf
          \State \Return "yes"
          \EndFunction
        \end{algorithmic}
      \end{algorithm}
      To prove the correctness of \Cref{alg:efInit}, we first prove that table $D$ has been filled correctly when we check if there is a fitting EF-init allocation for a "yes"-entry after \alglineref{alg:efInit:check}.
      \begin{claim}\label{claim:D_correct}
        After \alglineref{alg:efInit:check}, for any $a\in [m]$, $b\in [t^*]$, and $c\in [m]$, it holds that entry $D[a][b][c]=\text{"yes"}$ if and only if there is a fitting allocation for entry $D[a][b][c]$.
      \end{claim}
      \begin{claimproof}
        Note that for any $a \in  [m]$ and $b\in [t^*]$,  a fitting allocation can only exist if $c \leq k_b$. Moreover, in \Cref{alg:efInit}, for any $b\in [t^*]$, we only introduce $\text{"yes"}$-entries for $c \leq k_b$.
        Thus, for all $b\in [t^*]$, entries with $c>k_b$ are correctly set to "no", and it suffices to consider entries with $c \in [k_b]$.
        We prove this claim by induction over $b\in [t^*]$.

        First, we consider the base case of $b=1$. Note that for any $a\in [m]$ and $c \in [k_1]$, a fitting allocation $\mathcal{X}$ for $D[a][1][c]$ allocates exactly $a$ resources  only to agents in $L_1$, $|X_i|=c$ for every agent $i\in L_1$ and $\mathcal{X}$ satisfies that there is no envy between any pair $i,j \in L_1$ under EF-init.
        Clearly, there is such an allocation if and only if $|L_1|\cdot c =a$. As we set $D[a][1][c]=\text{"yes"}$  if and only if $|L_1|\cdot c =a$ in \alglineref{alg:efInit:base}, the base case holds.

        For the induction step, assume that for some $1\leq h<t^*$, the claim holds for all table entries $D[a][b][c]$ with $b \leq h$, $a \in [m]$ and $c \leq k_b$. We need to show that for any $a \in [m]$ and $c \in [k_{h+1}]$, entry $D[a][h+1][c]=\text{"yes"}$ if and only if there is a fitting allocation $\mathcal{X}$ for this entry.

        \paragraph{($\Leftarrow$):}
        Recall that we set $D[a][h+1][c]=\text{"yes"}$ in \alglineref{alg:efInit:combine} if we find a suitable $\text{"yes"}$-entry for the preceding level $L_h$.
        Suppose that there is a fitting allocation $\mathcal{X}$ for $D[a][h+1][c]$. In $\mathcal{X}$, every agent in level $h+1$ gets exactly $c$ resources and agents  in $A \setminus \bigcup_{\ell \in [h+1]} L_\ell$ do not receive any resources.
        Next, we will construct a fitting allocation $\mathcal{X'}$ for an entry in the previous level $h$.
        Clearly, by setting $X'_i=\emptyset$ for all $i \in  L_{h+1}$ and $X'_i=X_i$ for all $i \notin  L_{h+1}$, allocation $\mathcal{X}$ induces an allocation $\mathcal{X}'$ that allocates $a'=a-|L_{h+1}|\cdot c\geq 0$ resources and where the set of agents that get a resource is $\bigcup_{\ell \in [h]} L_\ell$. Furthermore, in $\mathcal{X}'$, there cannot be any envy under EF-init between agents $i,j \in \bigcup_{\ell \in [h]} L_\ell$, as these agents would also be envious in $\mathcal{X}$.
        Finally, by \Cref{EFinit:ordered}, in $\mathcal{X}'$, all agents in $L_h$ need to get the same number $c'$ of resources. It needs to hold that $c'\leq k_h$, since $\mathcal{X}$ is fitting.
        Thus, $\mathcal{X}'$ is fitting for $D[a'][h][c']$. By the induction assumption, this implies that $D[a'][h][c']=\text{"yes"}$.
        As in the algorithm, choose an agent $i^* \in L_{h}$ with minimal $v_{i^*}$ and $j^* \in L_{h+1}$ with maximal $v_{j^*}$, and let $c'_{min} = c + \lceil \frac{b_{j^*}-b_{i^*}}{v_{i^*}} \rceil$ and $c'_{max} = c + \lfloor \frac{b_{j^*}-b_{i^*}}{v_{j^*}} \rfloor$. Note that since ${i^*}$ does not envy ${j^*}$ under EF-init in $\mathcal{X}$, it needs to hold that
        \begin{align*}
          &b_{i^*} + u_{i^*}(X_{i^*}) \geq b_{j^*} + u_{i^*}(X_{j^*})\\
          \iff& b_{i^*} + v_{i^*} \cdot |X'_{i^*}| \geq b_{j^*} + v_{i^*} \cdot c\\
          \iff& |X'_{i^*}| -c \geq \frac{b_{j^*} - b_{i^*}}{v_{i^*}}\iff |X'_{i^*}| \geq c'_{min}.
        \end{align*}
        Furthermore, ${j^*}$ does not envy ${i^*}$ under EF-init in $\mathcal{X}$, so we have that
        \begin{align*}
          &b_{j^*} + u_{j^*}(X_{j^*}) \geq b_{i^*} + u_{j^*}(X_{i^*})\\
          \iff& b_{j^*} + v_{j^*} \cdot c \geq b_{i^*} + v_{j^*} \cdot |X'_{i^*}|\\
          \iff& \frac{b_{j^*}-b_{i^*}}{v_{j^*}} \geq |X'_{i^*}| -c \iff |X'_{i^*}| \leq c'_{max}.
        \end{align*}

        Note that as shown above, $c'_{min}\leq c' = |X'_{i^*}| \leq  c'_{max}$. Thus, we consider the "yes"-entry for $\mathcal{X}'$ in the algorithm and set entry $D[a][h+1][c]=\text{"yes"}$ in \alglineref{alg:efInit:combine}.

        \paragraph{($\Rightarrow$):}
        Conversely, assume that $D[a][h+1][c]=\text{"yes"}$. Let $a'=a-|L_{h+1}|\cdot c$.
        Again, choose $i^* \in L_{h}$ with minimal $v_{i^*}$ and $j^* \in L_{h+1}$ with maximal $v_{j^*}$, and  let $c'_{min} = c + \lceil \frac{b_{j^*}-b_{i^*}}{v_{i^*}} \rceil$ and $c'_{max} = c + \lfloor \frac{b_{j^*}-b_{i^*}}{v_{j^*}} \rfloor$.
        As we only set the entry to $\text{"yes"}$ in \alglineref{alg:efInit:combine}, there needs to exist some $c' \in \{c'_{min}, \dots, c'_{max}\}$ such that $D[a'][h][c']=\text{"yes"}$. By the induction assumption, this implies that there is a fitting allocation $\mathcal{X}$ for $D[a'][h][c']$. Extending $\mathcal{X}$ by allocating $c$ resources to every agent in level $h+1$ yields an allocation $\mathcal{X}'$ that allocates exactly $a'+c\cdot |L_{h+1}|=a$ resources and where the set of agents that receive a resource is $\bigcup_{\ell \in [h+1]} L_\ell$.
        It remains to prove that there is no envy under EF-init between any agents $i,j\in \bigcup_{\ell \in [h+1]} L_\ell$ in $\mathcal{X}'$. Clearly, there cannot be a pair of agents $i,j \in \bigcup_{\ell \in [h]} L_\ell$ where $i$ envies $j$, as $i$ would also envy $j$  (under EF-init) in $\mathcal{X}$. Next, consider two agents $i\in L_h$ and $j \in L_{h+1}$. Note that
        \begin{align*}
          &|X'_{i}|= c'\geq c'_{min}=c + \lceil \frac{b_{j^*}-b_{i^*}}{v_{i^*}} \rceil \geq |X'_{j}|+ \lceil \frac{b_{j}-b_{i}}{v_{i}} \rceil\\
          \Rightarrow& |X'_{i}| -|X'_{j}| \geq  \frac{b_{j}-b_{i}}{v_{i}} \\
          \iff& b_{i} + |X'_{i}| \cdot v_{i} \geq b_{j} + |X'_{j}|\cdot v_{i},
        \end{align*}
        so agent $i$ does not envy $j$ under EF-init. Furthermore, we have that
        \begin{align*}
          &|X'_{i}|= c'\leq c'_{max} = c + \lfloor \frac{b_{j^*}-b_{i^*}}{v_{j^*}} \rfloor \leq|X'_{j}| + \lfloor \frac{b_{j}-b_{i}}{v_{j}} \rfloor \\
          \Rightarrow & \frac{b_{j}-b_{i}}{v_{j}} \geq |X'_{i}| -|X'_{j}| \\
          \iff & b_{j} + |X'_{j}| \cdot v_{j} \geq b_{i} + |X'_{i}|\cdot v_{j},
        \end{align*}
        so agent $j$ does not envy $i$ under EF-init.
        It remains to prove that there is no envy under EF-init between an agent $i \in L_{h'}$ with $h'< h$ and an agent $j \in L_{h+1}$. Consider an agent $j' \in L_h$. We have that there is no envy between agents $i$ and $j'$, as there is no envy between any agents in $\bigcup_{\ell \in [h]} L_\ell$ (under EF-init), since $\mathcal{X}$ is fitting. Moreover, as argued above, there is no envy between agent $j' \in L_h$ and $j \in L_{h+1}$. With \Cref{lemma:EFInit_transitive}, it follows that there is no envy between agent $i \in L_{h'}$ with $h'< h$ and $j \in L_{h+1}$.
        Thus, $\mathcal{X}'$ is fitting for $D[a][h+1][c]$.
      \end{claimproof}

      Now, we are ready to prove the correctness of the algorithm in the following claim.
      \begin{claim}
        \Cref{alg:efInit} accepts if and only if there is a complete EF-init allocation.
      \end{claim}
      \begin{claimproof}
        We prove the claim by proving both directions of the equivalence.
        \paragraph{($\Rightarrow$):}
        First, suppose that \Cref{alg:efInit} accepts. Note that we only accept in \alglineref{alg:efInit:accept} if there is an entry $D[m][b][c]=\text{"yes"}$ for some $b\in [t^*]$ and $c \in [m]$. By \Cref{claim:D_correct}, there exists a fitting allocation $\mathcal{X}$ for $D[m][b][c]$. Clearly, this allocation is complete.

        Note that no agent in $A \setminus \bigcup_{\ell \in [b]} L_\ell$ gets a resource in $\mathcal{X}$. Thus, only agents in $\bigcup_{\ell \in [b]} L_\ell$ can be envied under EF-init by any agent. Since $\mathcal{X}$ is fitting for $D[m][b][c]$, no agent $i \in \bigcup_{\ell \in [b]} L_\ell$ can be envious of another agent under EF-init. Therefore, we only need to check if there is an agent $i\in A \setminus \bigcup_{\ell \in [b]} L_\ell$ that envies an agent $j\in \bigcup_{\ell \in [b]} L_\ell$ under EF-init.
        First, observe that no agent $i\in \bigcup_{\ell \in \{t^*,\dots, t\}} L_\ell$ envies an agent $j\in \bigcup_{\ell \in [b]} L_\ell$ under EF-init:
        Consider some agent $i\in \bigcup_{\ell \in \{t^*,\dots, t\}} L_\ell$ and an agent $j\in L_\ell$ for some $\ell \in [b]$. Since $\mathcal{X}$ is fitting, the allocation needs to satisfy the upper-bounds on the number of resources for agents in the first $t^*$ levels that we defined in the beginning. Formally, it needs to hold that $|X_j|\leq k_\ell$.
        Thus, we have that
        \begin{align*}
          |X_j|\leq k_\ell \leq \frac{b_i-b_j}{v_i} \iff b_i \geq b_j + v_i \cdot |X_j| \\\iff  b_i + u_i(X_i) \geq b_j + u_i(X_j),
        \end{align*}
        so agent $i$ does not envy $j$ under EF-init.
        Note that if $b=t^*$, we already have shown that there is no envy between any pair of agents under EF-init.
        Otherwise, if $b<t^*$, it remains to consider an agent $i\in \bigcup_{\ell \in \{b+1,\dots, t^*\}} L_\ell$. Let $j^*\in L_b$ be an agent in $L_b$.  Recall that $j^*$ does not envy any agent $j\in \bigcup_{\ell \in [b]} L_\ell$ under EF-init, since $\mathcal{X}$ is fitting. Thus, if $i$ does not envy $j^*$, by \Cref{lemma:EFInit_transitive}, $i$ does not envy $j$ (under EF-init). Note that $i$ does not envy $j^*$ under EF-init if and only if $c=|X_{j^*}|\leq \frac{b_{j^*}-b_i}{v_i}$, as
        \begin{align*}
             |X_{j^*}|\leq \frac{b_{j^*}-b_i}{v_i} \iff b_{j^*} + v_i \cdot |X_{j^*}| \leq b_i \\\iff b_{j^*} + u_i(X_{j^*})\leq b_i + u_i(X_i).
        \end{align*}
                 Since we accept in \alglineref{alg:efInit:accept} only if this inequality holds for all agents $i\in \bigcup_{\ell \in \{b+1,\dots, t^*\}} L_\ell$, allocation $\mathcal{X}$ is EF-init.

        \paragraph{($\Leftarrow$):}
        Suppose there is a complete EF-init allocation $\mathcal{X}$.
        If the algorithm accepts, we are done. In the following, assume for the sake of contradiction that the algorithm does not accept.
        Note that it needs to hold for all $h \in [t^*]$ and $i \in L_h$ that $|X_i|\leq k_{h}$ as argued in the beginning: If $t^*<t$, all agents in $\bigcup_{\ell \in \{t^* + 1,\dots,t\}} L_\ell$ cannot get a resource in  $\mathcal{X}$. Therefore, if $|X_i|> k_{h}$ for some $h \in [t^*]$ and $i \in L_h$, then there is some agent $j \in \bigcup_{\ell \in \{t^* + 1,\dots,t\}} L_\ell$ that envies $i$ under EF-init. If $t^*=t$, we set $k_{h}=m$ (recall that this is the case if there is no violating pair in the instance), so $|X_i|\leq k_{h}=m$ holds trivially.

        Let $h$ be the maximal level $h \in [t]$ such that an agent $i \in L_h$ receives a resource in $\mathcal{X}$. Note that $h\leq t^*$, since only agents in $\bigcup_{\ell \in [t^*]} L_\ell$ can get a resource in an EF-init allocation. Moreover, by \Cref{EFinit:ordered}, all agents in $\bigcup_{\ell \in [t]} L_\ell$ need to get at least one resource and every agent in $L_h$ needs to get exactly $|X_i|$ resources in $\mathcal{X}$. Clearly, since $\mathcal{X}$ is EF-init, there is no envy under EF-init between any pair of agents in $\mathcal{X}$. It follows that $\mathcal{X}$ is fitting for $D[m][h][|X_i|]$ and thus $D[m][h][|X_i|]=\text{"yes"}$ by \Cref{claim:D_correct}.

        If $h=t^*$, CheckEFInit($m, h, |X_i|$) returns "yes" and we accept in \alglineref{alg:efInit:accept}, which is a contradiction. Thus, assume that $h<t^*$ and consider an agent $i' \in \bigcup_{\ell \in \{h+1,\dots, t^*\}} L_\ell$ and an agent $j^*\in L_h$. As argued above, $i'$ does not envy $j^*$ under EF-init if and only if $|X_{j^*}|\leq \frac{b_{j^*}-b_{i'}}{v_{i'}}$. Thus, this inequality needs to hold for all agents $i' \in \bigcup_{\ell \in \{h+1,\dots, t^*\}} L_\ell$ in allocation $\mathcal{X}$. Therefore, CheckEFInit($m, h, |X_i|$) = "yes". Again, this implies that we accept in \alglineref{alg:efInit:accept}, which is a contradiction. Therefore, the algorithm needs to accept either for this entry or for an entry that was checked earlier.
      \end{claimproof}

      Finally, it remains to show that \Cref{alg:efInit} runs in $\mathcal{O}(n^2\cdot m^3)$ time and $\mathcal{O}(n\cdot m^2)$ space.
      Recall that we proceed in two steps: Firstly, we partition the agents into the levels, compute $t^*\in [t]$ and the upper-bounds $k_h$ for all $h\in [t^*]$. Secondly, we decide \EFEX{} using \Cref{alg:efInit}.

      Partitioning the agents into the levels and computing $t^*$ and the upper-bounds can be done in time $\mathcal{O}(n^2\cdot m^3)$. Moreover, we need to store at most $t^*\leq n$ upper-bounds.
      Next, we consider \Cref{alg:efInit}. First, since $t\leq n$, the table $D$ has size $\mathcal{O}(n\cdot m^2)$. Clearly, the initialization of the table within \alglineref{alg:efInit:initA} to \alglineref{alg:efInit:initB} can be done in running time $\mathcal{O}(n\cdot m^3)$.
      To compute a table entry for level $L_b$ with $b\in \{2,\dots, t\}$, we need to look at at most $m$ entries for level $L_{b-1}$.
      Note that we also need to find the agents $i \in L_{b-1}$ with minimal value $v_i$ and $j \in L_b$ with maximal value $v_j$. However, for every one of the at most $n$ levels, we can store the agents with minimal and maximal value when partitioning the agents into the levels.
      Thus, each one of the $\mathcal{O}(n\cdot m^2)$ entries can be computed in $\mathcal{O}(m)$, so the algorithm fills table $D$ in time $\mathcal{O}(n\cdot m^3)$.
      Lastly, we analyze the running time of checking if there is a "yes"-entry with a fitting EF-init allocation within the loop in \alglineref{alg:efInit:check}. First, observe that we check at most $\mathcal{O}(n\cdot m)$ entries. For every "yes"-entry, we check whether a fitting allocation is EF-init using function CheckEFInit, which takes time at most $\mathcal{O}(n\cdot m)$.
    \end{proof}
    Note that the running time and space consumption from above may be improved by a more careful analysis. For example, in \Cref{alg:efInit}, the space consumption of table $D$ can be improved by storing only the $\mathcal{O}(m^2)$ entries of the previous level. Whenever we introduce a $\text{"yes"}$-entry, we can directly check if a fitting allocation is also EF-init, in which case we can directly accept.

    Moreover, the algorithm can also be modified to compute an EF-init allocation by storing a fitting allocation for $\text{"yes"}$-entries. Note that it suffices to store an arbitrary fitting allocation for an entry: Extending a fitting allocation for an entry for level $L_b$ with $b \in [t]$ to level $L_{b+1}$ and checking if the allocation is EF-init only depends on the number of resources allocated to an agent in $L_b$, which is the same for all fitting allocations for this entry.

  A natural follow-up question to this result is whether it can be extended to decide \EFOEX{} for identical resources. However, directly adapting our dynamic programming approach seems challenging, since the key ``transitivity'' property from \Cref{lemma:EFInit_transitive} does not hold any more for EF1-init: For three agents $i,i^*,j \in A$, under EF1-init, we are allowed to disregard a resource each when checking whether $i$ envies $i^*$ and whether $i^*$ envies $j$. As a result,  even when there is no envy between $i$ and $i^*$ and between $i^*$ and $j$ under EF1-init, a single resource may not be sufficient to eliminate envy between $i$ and $j$. Moreover, it does not hold that agents in level $h^*$ (the first level containing an agent $j$ from a violating pair) and above cannot get a resource, if the differences in the utility agents from a violating pair derive for a resource are only small.

  \section{A Satisfiable Envy Notion}\label{sec:satisfiable}
  In \Cref{prop:EF1_init:nonex}, we showed that a complete EF1-init allocation may not exist, even in simple instances. This limits the practical applicability of EF1-init, since it does not provide guidance on which solution to select in many situations. Moreover, we proved that finding such an allocation is computationally intractable.
  Motivated by this, we propose a relaxation of EF1-init that is always satisfiable (see \Cref{sub:minEF1definition}) and show that an allocation satisfying this relaxed notion can be computed efficiently (see \Cref{sub:computation}). Furthermore, we identify a natural special case in which this relaxation coincides with EF1-init (see \Cref{pr:implication}).

  \subsection{Minimum EF1 with Initial Utilities}\label{sub:minEF1definition}

  To derive a meaningful fairness notion that is always satisfiable, we revisit \Cref{prop:EF1_init:nonex}. This counterexample illustrates that EF1-init allocations fail to exist when there is an inherent conflict between resolving the disadvantage of initially worse-off agents and treating initially better-off agents in a way that, from their perspective, is fair: agent $j$ with lower initial utility also derives a lower utility from each resource than agent $i$, so the agents fundamentally disagree on how many resources are needed to bridge their initial utility gap. Thus, there will always be an agent who is envious under EF1-init.
  To resolve this conflict, we deviate from the comparison approach taken by EF1-init where we compare initial utilities plus bundle values. Instead, when checking whether an agent $i$ with higher initial utility envies another agent $j$ with lower initial utility,  we first allow agent $j$ to receive a subset of resources $X^* \subseteq X_j$ bounded relative to their initial utility gap $b_i-b_j$. Intuitively,  following the principle of \emph{equality of outcome}, the purpose of the set $X^*$ is to allow a differentiated treatment of initially worse-off agents to resolve their initial disadvantage. We then disregard $X^*$ when evaluating if the higher-utility agent $i$ is envious of $j$'s bundle.
  In the following, we start by presenting a first implementation of this idea, which, however, turns out to be not always satisfiable, and then define minimum-EF1-init, our satisfiable envy notion. 
  
  \paragraph{A first relaxation attempt}
  Formally, one could adjust the EF1-init definition for the case $b_i > b_j$ as follows:
  \begin{align*}
    &\text{If }b_i> b_j, \text{ there is } X^* \subseteq X_j \text{ with } u_j(X^*)< b_i-b_j \text{ and } r\in X_j \\
    &\text{so that }u_i(X_i) \geq u_i(X_j \setminus (X^* \cup \{r\}))\text{.} \tag{1} \label{eq:1}
\end{align*}

  However, this relaxation is not satisfiable if there are multiple agents with low initial utility values that have very different utility functions:
  \begin{example}\label{ef1iw:nonex}
    Consider an instance with agents $A = \{1,2,3\}$ and $m=10$ identical resources. The agents have initial utilities $b_1 = b_2 = 0$ and $b_3 = 20$. Agents $1$ and $3$ have $u_1({\{r\}}) = u_3(\{r\}) = 20$ for every $r \in R$, while agent $2$ has $u_2(\{r\}) = 5$.  Since agents $1$ and $2$ have the same initial utility, we need to have $|X_2|-|X_1|\leq 1$. Moreover, \Cref{eq:1} requires that $|X_1|-|X_3|\leq 1$, since only the set $X^*=\emptyset$ satisfies $u_1(X^*) < 20 - 0$. However, this implies that $2$ envies $3$.
  \end{example}

  \paragraph{A satisfiable relaxation}
  Intuitively, \Cref{ef1iw:nonex} shows that it is insufficient to only consider each pair of agents individually when restricting the set $X^*$ in case the difference in initial utility can be offset more easily for some agents than others.
  To overcome this, when checking whether an agent $i$ with higher initial utility envies an agent $j$ with lower initial utility, we relax the restriction on $X^*$ by considering, for each resource in $X^*$, the minimum utility it gives to any agent with a lower initial utility than $i$.
  This yields the following notion:

  \begin{definition}\label{def:ef1iwmin}
    An allocation $\mathcal{X}$ is \emph{minimum-EF1-init} (min-EF1-init) if for every pair $i,j \in A$, we have $X_j=\emptyset$\footnote{As in \Cref{def:EFinit}, this special case is necessary to ensure that min-EF1-init can be satisfied on instances with large initial utility disparities (by an allocation where high-initial-utility agents receive no resources). Intuitively, perceiving such an allocation as desirable is justified since it is impossible to make the allocation fairer from the perspective of the agents with lower initial utility by taking resources away from the agents with higher initial utility.} or it holds that:
    \begin{itemize}
      \item[(C1)] If $b_i \leq b_j$, then  $b_i + u_i(X_i) \geq b_j + u_i(X_j\setminus \{r\})$ for some $r\in X_j$.
      \item[(C2)] If $b_i> b_j$, then there exists a resource $r\in X_j$ and a subset $X^* \subseteq X_j$ with
        \[
          \sum_{r' \in X^*} \min_{j' \in A \colon b_{j'} < b_i} u_{j'}(\{r'\}) < b_i-b_j,
        \]
        such that $u_i(X_i) \geq u_i(X_j \setminus (X^* \cup \{r\}))$ holds.
    \end{itemize}
  \end{definition}
  Intuitively, min-EF1-init ensures fairness in the following way.  When comparing an agent to another agent who initially is at least as well off, condition (C1) ensures that the disadvantage of the initially worse-off agent is compensated through the resources allocated to them.
  Conversely, when considering whether some initially better-off agent $i$ envies an initially worse-off agent $j$, we require condition (C2), which allows $j$ to receive a carefully bounded set $X^*$ of additional resources to counter $j$'s initial disadvantage compared to $i$. After the initial utility difference has been accounted for by the set $X^*$, condition (C2) ensures that the remaining resources are distributed in a way that is perceived as fair by agent $i$.
  
  Let us now take a closer look at the definition of $X^*$ in (C2), whose definition follows the principle of equality of outcome in the following way.
We restrict the set $X^*$ by considering the minimum utility of an agent with initial utility less than $b_i$ for each resource. Recalling \Cref{ef1iw:nonex}, this is necessary since considering simply the utility of $j$ as in \Cref{eq:1} may not be satisfiable: in \Cref{ef1iw:nonex}, it implies $|X_2|-|X_3|\leq 2$, violating (C1) as agent $2$ does not reach $3$'s initial utility. On the other hand, (C2) allows for an allocation where $|X_2|-|X_3|= 4$, compensating $2$'s initial disadvantage.
Note that the alternative way to overcome the nonexistence in \Cref{ef1iw:nonex}  would be to relax condition (C1) such that it is satisfied by an allocation with $|X_2|-|X_3|\leq 2$, which would then allow for stricter restrictions on the set $X^*$ as in \Cref{eq:1}. However, in such an allocation, even the post-allocation utility of agent $2$ would be smaller than the initial utility of $3$, violating equality of outcome.

    Revisiting \Cref{prop:EF1_init:nonex}, where no complete EF1-init allocation exists, one can observe that in this instance an allocation $\mathcal{X}$ is min-EF1-init if and only if $|X_i|=1$ and $|X_j|=3$.
    This allocation is desirable since it resolves agent $j$'s initial disadvantage, but does not overcompensate (i.e., does not give $j$ "too many" resources) due to the bound on the set $X^*$: Note that assigning all four resources to agent $j$ violates min-EF1-init, since only a set $X^*$ with $|X^*|=2$ fulfills the condition in (C2), so $|X_j|-|X_i|\leq3$ in any min-EF1-init allocation.
    
    While we believe no substantially stronger notions than min-EF1-init are satisfiable in this setting,\footnote{\label{fn:suggestion}A potential strengthening is to replace the restriction on $X^*$ in (C2) with $\min_{j' \in A : b_{j'} < b_i} u_{j'} (X^*) < b_i - b_j$; whether this notion is always satisfiable remains an open question.} it is also clearly a weaker requirement than \Cref{eq:1} and may permit allocations that agents with high initial utilities could perceive as unfair.\footnote{We remark that the allocation discussed in \Cref{ex:shortcomings} is not the unique min-EF1-init allocation for this instance, and there exist more balanced min-EF1-init allocations.}
    \begin{example}\label{ex:shortcomings}
      Consider three agents $A = \{1,2,3\}$ with $b_1 = b_2 = 0$, $b_3 = 10$, and $m = 100$ resources. Agent $1$ has $u_1(\{r^*\}) = 500$ for a particular $r^* \in R$ and $u_1(\{r\}) = 0$ for all other $r\in R$. Agents $2$ and $3$ have $u_2(\{r\}) = u_3(\{r\}) = 50$ for all $r\in R$. The allocation $X_1 = \{r^*\}$, $X_2 = R \setminus \{r^*\}$, $X_3 = \emptyset$ is min-EF1-init. Agents $1$ and $2$ do not envy agent $3$ since $X_3=\emptyset$. Agent $3$ does not envy agent $1$ since $|X_1|=1$, and $3$ does not envy agent $2$ since $u_1(X_2)=0$ allows to choose $X^*=X_2$ in (C2). Further,  it is easy to see that agent $1$ does not envy agent $2$ and vice versa. However, agent $3$ may regard such an allocation as unfair.
    \end{example}
    However, as discussed above, alternative ways of resolving the conflict in \Cref{ef1iw:nonex} lead to outcomes that are unacceptable from the perspective of equality of outcome, since they may fail to resolve the disadvantage of initially worse-off individuals.  We therefore believe that this tradeoff is unavoidable to achieve a satisfiable envy-based fairness notion. 
    On the positive side, we show that min-EF1-init implies EF1-init if agents' utility for a resource diminishes in their initial utility.
    \newcommand{\diminishing}{diminishing}
    \newcommand{\dimin}{dim}
    \begin{definition}\label{def:dimin}
      We call utilities \emph{\diminishing}, if $b_i<b_j$ implies  $u_i(\{r\})\geq u_j(\{r\})$ for all $i,j\in A$ and $r\in R$.
    \end{definition}

    While diminishing utilities can be a strong assumption when there are many non-identical resources and when the agents have heterogeneous, subjective utilities over them, they capture a meaningful and practically relevant class of scenarios where utility is determined primarily by an agent’s initial condition, rather than by individual preferences. For instance, in aid-based applications,  an agent’s initial utility may represent their current state or need. Agents in worse initial states typically benefit more from support resources: In the ARMMAN program featured in the introduction, a service call to a disengaged beneficiary can substantially increase engagement, whereas such a call has a limited effect on someone already very actively engaged.

    \begin{proposition}\label{pr:implication}
      Let $\mathcal{X}$ be an allocation that is min-EF1-init. If utilities are diminishing, then allocation $\mathcal{X}$ is EF1-init.
    \end{proposition}
    \begin{proof}
      In the following, to highlight that an inequality follows from \Cref{def:dimin}, we will write $\geq_{\text{\dimin}}$ and $\leq_{\text{\dimin}}$, and analogously $\leq_{\text{(C2)}}$ to highlight that an inequality follows from (C2) in \Cref{def:ef1iwmin}.
      We need to verify that for every pair of agents $i,j \in A$, either $X_j = \emptyset$ or there exists a resource $r\in X_j$ such that
      \[
        b_i + u_i(X_i) \geq b_j + u_i(X_j \setminus \{r\}). \label{eq:ef1} \tag{$*$}
      \]
      Since $\mathcal{X}$ is min-EF1-init,  conditions (C1) and (C2) from \Cref{def:ef1iwmin} hold.
      If $b_i \leq b_j$, (C1) implies \eqref{eq:ef1}. If $b_i > b_j$, we have that (C2) holds for some set $X^* \subseteq X_j$ and resource $r\in X_j$. Note that
      $
      u_i(X^*) = \sum_{r' \in X^*} u_i(\{r'\}) \leq_{\text{\dimin}}  \sum_{r' \in X^*} \min_{j' \in A \colon b_{j'} < b_i} u_{j'}(\{r'\}) <_{\text{(C2)}} b_i-b_j.
      $
      Since  $u_i(X_i) \geq u_i(X_j \setminus (X^* \cup \{r\}))$ by (C2), it follows that
      $
      u_i(X_j \setminus \{r\}) \leq u_i(X_j \setminus (X^* \cup \{r\})) + u_i(X^*) \leq_{\text{(C2), \dimin}} u_i(X_i) +  b_i - b_j,
      $
      which implies \eqref{eq:ef1}.
    \end{proof}

    \subsection{A Round-Robin Algorithm Guaranteeing Minimum EF1 with Initial Utilities}\label{sub:computation}

    In this section, we propose an algorithm that computes a complete min-EF1-init allocation. For this, we extend the well-known ``round-robin'' algorithm for the setting with initial utilities.\footnote{In the standard round-robin algorithm, we fix an ordering of the agents. In each round of the algorithm, following this ordering, agents pick their favorite so far unallocated resource. Once there are no more resources left, the algorithm terminates. In the classic setting without initial utilities, the resulting allocation is guaranteed to satisfy EF1.}

      \begin{algorithm}[t]
    \begin{algorithmic}[1]
      \setlength{\itemsep}{0.2em} 
      \State Partition set of agents $A$ into levels $L_1, \dots, L_t$ and initialize $\mathcal{X}$ as the empty allocation.
      \State Initialize set of \emph{active} agents $\mathcal{L} \subseteq A$ to $L_1$.
      \State Initialize linear order $\preceq$ over $\mathcal{L}$ to some arbitrary linear order over $L_1$.\protect\footnotemark
      \While{there are unallocated resources}
      \State{\textbf{Start a new \emph{round}:}}
      \While{an agent in $\mathcal{L}$ has not picked in current round}
      \LongState{Pick first agent $i\in \mathcal{L}$ according to $\preceq$ that has not picked in this round.}
      \LongState{Agent $i$ picks unallocated resource $r\in R$ with maximum utility $u_i(\{r\})$: $X_i \gets X_i \cup \{r\}$.}
      \If{all resources are allocated} \Return $\mathcal{X}$\EndIf
      \LongElsIfState{all $i' \in \mathcal{L}$ reached the agents in some\\ level $L_\ell$ with $\ell \in [t]$ and $L_\ell \cap \mathcal{L}=\emptyset$}
      \LongState[1em]{Let $L_\ell$ be the level with minimum $\ell \in [t]$\\ that satisfies the stated condition.}
      \LongState[1em]{$\mathcal{L} \gets \mathcal{L} \cup L_\ell$}
      \LongState[1em]{Extend $\preceq$ (maintaining the order of agents\\ in $\mathcal{L} \setminus L_\ell$) to the updated $\mathcal{L}$ such that:}
      \LongStatex[5.5em]{$\triangleright$ all agents that have already picked in this round are before the agents in $L_\ell$,}
      \LongStatex[5.5em]{$\triangleright$ all agents in $\mathcal{L} \setminus L_\ell$ that have not yet picked are after the agents in $L_\ell$.}
      
      \EndWhile
      \EndWhile
    \end{algorithmic}
    \caption{Round-Robin with Initial Utilities}
    \label{alg:rr_init}
  \end{algorithm}

    \paragraph{Idea}
    In our algorithm, we use the partitioning of the agents into $t$ levels by their initial utility, as introduced in \Cref{def:levels}.
    Our extension is based around the idea that at any point of our round-robin algorithm, only a subset of the agents is \emph{active} and can be assigned a pick. We maintain a picking order over all currently active agents in the algorithm. In the beginning, only the agents with the lowest initial utility in level $L_1$ are active. After each pick, we check whether all currently active agents have \emph{reached} the agents in some (not yet active) level $L_\ell$ with initial utility value $b$. By this, we mean that for each currently active agent $i$, the sum of $i$'s initial utility and $i$'s utility for their bundle is at least $b$ (all active agents believe that they have ``reached the initial situation'' of all agents from $L_\ell$).
    \begin{definition}\label{def:algo:rr_init:reached}
      We say that an agent $i\in A$ with bundle $X_i \subseteq R$ has \emph{reached} an agent $j \in A$ if $b_i + u_i(X_i) \geq b_j$.
    \end{definition}
    If all active agents have reached the agents in $L_\ell$, the agents in $L_\ell$ become active and are inserted into the picking order after the agents that have already picked in this round, meaning that they get to pick directly after being activated. This ensures that each agent from $L_\ell$ prefers the resources picked by them over the resources picked by any other agent that was active before them after $L_\ell$'s activation.
    Another observation crucial to the inner workings of the algorithm is that if agent $j$ is activated after another agent $i$ picked some resource $r$, then $i$ had not yet reached $j$ in the round before, implying that the value of $i$ for their current bundle without $r$ is less than the initial utility difference $b_j -b_i$.
    A full description is given in \Cref{alg:rr_init}.
    We show that \Cref{alg:rr_init} returns a complete allocation satisfying min-EF1-init, implying that a min-EF1-init allocation always exists.
    \begin{restatable}{theorem}{MainTheorem}\label{main-theorem}
      The allocation computed by \Cref{alg:rr_init} is complete and satisfies min-EF1-init. Further,  \Cref{alg:rr_init} runs in polynomial time.
    \end{restatable}
    
    \footnotetext{We write $i \preceq j$ to denote that $i$ picks before $j$.}
    \subsection{Proof of \texorpdfstring{\Cref{main-theorem}}{Theorem~\ref{main-theorem}}}
    In this subsection, we prove \Cref{main-theorem}. We start with defining notation and formalizing two observations (\Cref{subsub:1}) before we present the crucial activation gap lemma (\Cref{subsub:2}) and finally the proof of  \Cref{main-theorem} (\Cref{subsub:3}).

    \subsubsection{Notation and Initial Observations}\label{subsub:1}
    An agent $i \in A$ is \emph{active} in a round if $i \in \mathcal{L}$ at the end of the round. An agent $i \in A$ is \emph{activated} in a round (after/by a pick) if $i$ is added to  $\mathcal{L}$ during the round (after/by the pick).
    A \emph{picking sequence} starting with/after some pick is the contiguous subsequence of all picks until the end after this pick.\footnote{We will also consider picking sequences starting with a given pick up to another pick, by which we mean the contiguous subsequence of all picks in the algorithm between the two given picks.}  
    We first observe the following properties of \Cref{alg:rr_init}.
\begin{observation}\label{alg:rr_init:active}
  Consider \Cref{alg:rr_init} and let $\preceq$ be the final linear order when the algorithm terminates.
  Let $i,j \in A$ be two distinct agents. For $k \in \{i,j\}$, let $X^*_k$ be the bundle picked by $k$ in a given subset $\mathcal{R}$ of rounds in which both $i$ and $j$ are active and where $j$ picks in at least one round in $\mathcal{R}$.
  Then it holds that:
  \begin{enumerate}
    \item  Assume that $i$ picks in the last round\footnote{The ordering of the rounds in  $\mathcal{R}$ that we refer to here is given by the order in which they occur in \Cref{alg:rr_init}.} in $\mathcal{R}$. For the resource $r \in X^*_i$ picked by agent $i$ in the last round in $\mathcal{R}$ and some $r' \in X^*_j$, we have that $u_i(X^*_i \setminus \{r\}) \geq u_i(X^*_j \setminus \{r'\})$.
    \item If $i$ does not pick in the last round in $\mathcal{R}$, we have that $u_i(X^*_i) \geq u_i(X^*_j \setminus \{r'\})$ for some $r' \in X^*_j$.
    \item If $i \preceq j$, then $u_i(X^*_i) \geq u_i(X^*_j)$.
  \end{enumerate}
\end{observation}
Secondly, we formalize an observation from above:
\begin{observation}\label{alg:rr_init:notreached}
  Consider \Cref{alg:rr_init} and let $i,j \in A$ be two agents with $b_i > b_j$ such that $i$ is activated after agent $j$ picked some resource $r\in R$. Let $X^*_j$ be the bundle of $j$ after this pick.
  Then, it holds that $u_j(X^*_j \setminus \{r\}) < b_i - b_j$.
\end{observation}

\subsubsection{The Activation Gap Lemma}\label{subsub:2}
    The following key lemma gives a bound on the value of the bundle assigned to an active agent $j$ at the point when some agent $i$ with a higher initial utility is activated. More concretely, we show that the sum of the minimum utility of an agent with lower initial utility than $i$ for each resource in $j$'s bundle is at most the initial utility gap between $i$ and $j$ after removing some resource $r$.  This is crucial in our proof of \Cref{main-theorem}, since it implies that $j$'s current bundle without $r$ satisfies the restriction on $X^*$ in min-EF1-init (C2), allowing us to disregard the resources picked by $j$ in rounds when $i$ was not active when checking whether $i$ envies $j$.

    \begin{restatable}[Activation Gap Lemma]{lemma}{ActiveMinLemma}\label{alg:rr_init:activeMin}
      Consider \Cref{alg:rr_init} and let $i,j\in A$ be two agents that were active in some round with $b_i \geq b_j$ and $b_i > \min_{i'\in A} b_{i'}$.
      For $k \in A$, let $X^*_k$ be the bundle of $k$ after the pick that activated $i$. Then, either $X^*_j=\emptyset$ or there exists $r \in X^*_j$ such that
      \[
        \sum_{r' \in X^*_j \setminus \{r\}} \min_{j' \in A \colon b_{j'} < b_i} u_{j'}(\{r'\}) < b_{i} - b_j.
      \]
    \end{restatable}
   As before, we define a picking sequence starting with/after some pick as the contiguous subsequence of all picks until the end of the algorithm after this pick. In the following, we will also consider picking sequences starting with a given pick up to another pick, by which we mean the contiguous subsequence of all picks in the algorithm between these two given picks. We first observe the following property.
    \begin{observation}\label{alg:rr_init:upToOneThenBefore}
      Consider \Cref{alg:rr_init} and let $\mathcal{X}$ be the allocation returned by the algorithm. Let $i, j \in A$ be two distinct agents, and for each $k \in \{i, j\}$, let $X^*_k$ denote the bundle of $k$ after a given pick in the algorithm.
      If either $X^*_j = \emptyset$ or there exists a resource $r \in X^*_j$ such that $b_i + u_i(X^*_i)\geq  b_j + u_i(X^*_j \setminus \{r\})$, and after this pick agent $i$ picks before $j$'s next pick, then it holds that $b_i + u_i(X_i)\geq  b_j + u_i(X_j \setminus \{r\})$ in the final allocation $\mathcal{X}$.
    \end{observation}
    Now, we are ready to prove the crucial Activation Gap Lemma.
    \begin{proof}
      For agents $a,b \in A$ with $a\in L_h$ and $b\in L_{h'}$ for $1\leq h'\leq h \leq t$, we define $\Delta_{a,b} \coloneqq h -h'$.
      We prove the claim by induction over $\Delta_{i,j}$. For $\Delta_{i,j}=0$, the claim holds trivially, as both agents are activated in the same round and thus $|X^*_j|\leq 1$.

      \paragraph{Induction Step.} Assume that the claim holds for all $i,j \in A$ with $b_i\geq b_j$ and $b_i > \min_{i'\in A} b_{i'}$,  and with $\Delta_{i,j} \leq \Delta^*$ for some $\Delta^* \in  [t-1]$. We show that the claim then also holds for $i,j \in A$ with $b_i\geq b_j$ and $b_i > \min_{i'\in A} b_{i'}$,  and with $\Delta_{i,j} = \Delta^*+1$.
      If $X^*_j=\emptyset$, then the claim holds trivially, so we assume in the following that $X^*_j\neq\emptyset$.
      Note that we have that $b_{i} > b_j$ since $\Delta_{i,j} \geq 1$. Let $j^* \in A$ be the agent after whose pick agent $i$ was activated and let $r\in X^*_{j^*}$ be the resource picked by $j^*$ in this pick.
      Let $\preceq$ be the final linear order when the algorithm terminates.
      If $j = j^*$, the claim holds, since $b_{i} - b_j > u_{j}(X^*_j \setminus \{r\})$  for the resource $r \in X^*_j$ picked by $j$ in the round when $i$ was activated by \Cref{alg:rr_init:notreached}
      and $u_{j}(X^*_j \setminus \{r\}) \geq  \sum_{r' \in X^*_j \setminus \{r\}} \min_{j' \in A \colon b_{j'} < b_i} u_{j'}(\{r'\})$ since $b_j<b_i$.

      Thus, assume that $j^* \neq j$. Note that we have that $j^* \preceq i$. We make a case distinction based on the initial utility values $b_j$ and $b_{j^*}$.

      \paragraph{Case 1 ($b_j=b_{j^*}$).}
      First, consider that $b_j=b_{j^*}$. Note that this implies that $j$ and $j^*$ are active in the same rounds.
      Let $\mathcal{R}$ be the set of rounds in which $j$ and $j^*$ are active up to (and including) the round in which $i$ is activated. In these rounds, agent $j^*$ picks the resources in $X^*_{j^*}$ and agent $j$ picks a superset of the resources $X^*_{j}$ (since agent $j$ may pick after $i$ is added in the last round in $\mathcal{R}$).
      \Cref{alg:rr_init:active}~(1) implies that $u_{j^*}(X^*_{j^*} \setminus \{r\})\geq u_{j^*}(X^*_{j} \setminus \{r'\})$ for $r \in X^*_{j^*}$ picked last by $j^*$ and some $r' \in X^*_{j}$.
      Moreover, it holds that $u_{j^*}(X^*_{j^*} \setminus \{r\}) < b_{i} - b_j$ by \Cref{alg:rr_init:notreached}. Thus, the claim holds in this case since $b_i -b_j > u_{j^*}(X^*_{j} \setminus \{r'\}) \geq  \sum_{r'' \in X^*_j \setminus \{r'\}} \min_{j' \in A \colon b_{j'} < b_i} u_{j'}(\{r''\})$.

      \paragraph{Case 2 ($b_j > b_ {j^*}$).}
      Next, consider that $b_j > b_ {j^*}$ and let $\tilde{X^*_{j^*}}$ be the bundle of $j^*$ after the pick that activated $j$. Since $j$ was activated, it needs to hold that
      \begin{align}\label{app:claimUpToOne:jjp}
        u_{j^*}(\tilde{X^*_{j^*}}) \geq b_j - b_{j^*}.
      \end{align}
      Recall that agent $j^*$ picked the resource $r \in X^*_{j^*}$ in the round in which $i$ is activated. Thus, the bundle of agent $j^*$ in the round prior is $X^*_{j^*} \setminus \{r\}$.
      By \Cref{alg:rr_init:notreached}, we get that $u_{j^*}(X^*_{j^*} \setminus \{r\})< b_i-b_{j^*}$.
      Note that $r \notin \tilde{X^*_{j^*}}$, as $\tilde{X^*_{j^*}}$ is the bundle of $j^*$ when $j$ was activated and $r$ is picked by $j^*$ when $i$ is activated, and it holds that $b_i>b_j$.
      Using \Cref{app:claimUpToOne:jjp}, this implies that
      \begin{align}\label{app:claimUpToOne:ij}                                                                          u_{j^*}(X^*_{j^*} \setminus (\tilde{X^*_{j^*}} \cup \{r\}))< b_i-b_{j},
      \end{align}
      Note that it suffices to show that
      \begin{align}\label{app:eq:minInd:goal}
        b_i - b_{j} > u_{j^*}(X^*_j \setminus \{r'\})
      \end{align}
      for some $r' \in X^*_j$, since  $$u_{j^*}(X^*_j \setminus \{r'\}) \geq \sum_{r'' \in X^*_j \setminus \{r'\}} \min_{j' \in A \colon b_{j'} < b_i} u_{j'}(\{r''\}),$$ as $b_{j*}<b_i$.
      To show that there exists a resource such that \Cref{app:eq:minInd:goal} holds, we now make a case distinction between $j \preceq j^*$ and $j^* \preceq j$.
      Note that $j^* \preceq j$ implies that $j^* \preceq i \preceq j$, since $b_j < b_i$ and there can only be agents with initial utility at least $b_i$ between $j^*$ and $i$ in $\preceq$: Since $i$ was activated directly after $j^*$'s pick, only agents that became active after or by the same pick as $i$ can be between $j^*$ and $i$.

      If $j^* \preceq i \preceq j$, consider the set $\mathcal{R}$ of all rounds after the round in which $j$ is activated up to (excluding) the round in which $i$ is activated. The resources picked by agent $j^*$ in the rounds in $\mathcal{R}$ is the set $X^*_{j^*} \setminus (\tilde{X^*_{j^*}} \cup \{r\})$.
      Note that agent $j$ has not yet picked in the round when $i$ is activated, since $j^* \preceq i \preceq j$. Thus, the resources picked by $j$ in the rounds in $\mathcal{R}$ is exactly the set $X^*_j \setminus \{r'\}$ for the resource $r' \in X^*_j$ picked first by $j$.
      Since $j^* \preceq j$, we can apply \Cref{alg:rr_init:active}~(3) to show that \Cref{app:eq:minInd:goal} holds for resource $r'$, as \[
        b_i-b_{j}>_{\eqref{app:claimUpToOne:ij}} u_{j^*}(X^*_{j^*} \setminus (\tilde{X^*_{j^*}} \cup \{r\}))\geq_{\eqref{alg:rr_init:active}} u_{j^*}(X^*_j \setminus \{r'\}).
      \]

      If $j \preceq j^*$, we consider the set $\mathcal{R}$ of all rounds in which $j$ is active up to (now including) the round in which $i$ is activated. Note that as $j \preceq j^*$, agent $j^*$ has not yet picked in the round when $j$ is activated, so $\tilde{X^*_{j^*}}$ only contains resources picked by $j^*$ in rounds in which $j$ is not active. Thus, the resources picked by agent $j^*$ in the rounds in $\mathcal{R}$ is the set $X^*_{j^*} \setminus \tilde{X^*_{j^*}}$ and agent $j$ picks the resources in $X^*_j$.
      Using \Cref{alg:rr_init:active}~(1), it follows there exists a resource $r' \in X^*_j$ such that \Cref{app:eq:minInd:goal} holds, since
      \[
        b_i-b_{j}>_{\eqref{app:claimUpToOne:ij}} u_{j^*}(X^*_{j^*} \setminus (\tilde{X^*_{j^*}} \cup \{r\}))\geq_{\eqref{alg:rr_init:active}} u_{j^*}(X^*_j \setminus \{r'\}),
      \]
      for resource $r' \in X^*_j$ from \Cref{alg:rr_init:active}~(1).
      Thus, the claim holds if $b_j > b_{j^*}$.

      \paragraph{Case 3 ($b_j < b_{j^*}$).}
      It remains to consider that $b_j < b_{j^*}$. Note that this implies that $b_{j^*}$ is strictly greater than the minimum initial utility of all agents. Since $b_{j^*} < b_{i}$, we have that $\Delta_{j^*, j}\leq \Delta^*$, so we can apply the induction hypothesis for $j^*$ and $j$. Let $\tilde{X^*_{j}}$ be the bundle of $j$ after the pick that activated $j^*$.
      Note that it cannot be the case that $\tilde{X^*_{j}}= \emptyset$, since  $b_j < b_{j^*}$.
      Thus, there exists some resource $r' \in \tilde{X^*_{j}}$ such that
      \begin{align}\label{app:claimUpToOne:ih}
        b_{j^*} - b_j > \sum_{r'' \in \tilde{X^*_{j}} \setminus \{r'\}} \min_{j' \in A \colon b_{j'} < b_{j^*}} u_{j'}(\{r''\}).
      \end{align}

      Moreover, by \Cref{alg:rr_init:notreached}, we have that $u_{j^*}(X^*_{j^*} \setminus \{r\})< b_i - b_{j^*}$ for the resource $r \in X^*_{j^*}$ picked last by $j^*$.
      Consider the picking sequence starting with agent $j^*$'s first pick up to (excluding) the pick by $j^*$ that activates $i$. Agent $j$ picks the resources in $X^*_{j} \setminus \tilde{X^*_{j}}$, and agent $j^*$ the resources in $X^*_{j^*} \setminus \{r\}$.
      Since agent $j^*$ picks first in the picking sequence, it follows that $u_{j^*}(X^*_{j^*} \setminus \{r\}) \geq u_{j^*}(X^*_{j} \setminus \tilde{X^*_{j}})$, so we have that
      \begin{align}\label{app:claimUpToOne:ijph}
        b_{i}-b_{j^*} >_{\eqref{alg:rr_init:notreached}} u_{j^*}(X^*_{j^*} \setminus \{r\}) \geq u_{j^*}(X^*_{j} \setminus \tilde{X^*_{j}}).
      \end{align}
      This implies the claim, since
      \begin{align*}
        &b_i - b_j = (b_i - b_{j^*}) + (b_{j^*} - b_j) \\
        >_{\eqref{app:claimUpToOne:ijph},\eqref{app:claimUpToOne:ih}} &u_{j^*}(X^*_{j} \setminus \tilde{X^*_{j}}) + \sum_{r'' \in \tilde{X^*_{j}} \setminus \{r'\}} \min_{j' \in A \colon b_{j'} < b_{j^*}} u_{j'}(\{r''\}) \\
        \geq &\sum_{r'' \in X^*_{j} \setminus \{r'\}} \min_{j' \in A \colon b_{j'} < b_{i}} u_{j'}(\{r''\}).
      \end{align*}
    \end{proof}

\subsubsection{Putting the Pieces Together}\label{subsub:3}
    We are now ready to  prove that \Cref{alg:rr_init} computes a complete min-EF1-init allocation. Intuitively, an agent $i$ does not envy another agent $j$ with higher initial utility, since $i$ is allowed to pick resources until $i$ has reached $j$ before $j$ gets activated and makes their first pick. Moreover, $i$ does not envy an agent with lower initial utility, since the bound from \Cref{alg:rr_init:activeMin} allows us to disregard the resource picked by this agent before $i$ was activated in min-EF1-init, and agent $i$ picks before the other agent in the picking sequence after this pick.
   \begin{proof}[Proof of \Cref{main-theorem}]
      Let $\mathcal{X}$ be the allocation computed by the algorithm and let $\preceq$ be the final linear order. We show that for any two agents $i,j \in A$, agent $i$ does not envy agent $j$ under min-EF1-init in $\mathcal{X}$. Thus, we need to show that $X_j = \emptyset$ or condition (C1) holds if $b_i\leq b_j$ or (C2) holds if $b_i>b_j$. Therefore, in the following, assume that $X_j\neq \emptyset$.

      \paragraph{Case 1 ($b_i=b_j$).} In this case, we need to argue that there is a resource $r\in X_j$ such that condition (C1) holds. Note that $i$ and $j$ are active in the same rounds. Thus, condition (C1) holds by applying \Cref{alg:rr_init:active} on the set of all rounds: At least one of the cases is applicable, and every case implies that $u_i(X_i) \geq u_i(X_j \setminus \{r\})$ for some $r\in X_j$.

      \paragraph{Case 2 ($b_i < b_j$).} For $k \in A$, let $X^*_k$ be the bundle of $k$ after $j$'s first pick and let $r\in R$ be the first resource picked by $j$, i.e.,  $X^*_j= \{r\}$. Since $i$ had reached $j$ when agent $j$ was activated, it holds that
      \[
        b_i + u_i(X^*_i) \geq b_j + \underbrace{u_i(X^*_j \setminus \{r\})}_{=0}.
      \]
      Moreover, since $i$ picks before $j$ in the picking sequence starting with the pick after $j$'s first pick, agent $i$ prefers the resources they pick in this picking sequence (i.e., $X_i\setminus X^*_i$) over those picked by $j$ (i.e., $X_j\setminus X^*_j$). Thus, $u_i(X_i\setminus X^*_i)\geq u_i(X_j\setminus X^*_j)$. With \Cref{alg:rr_init:upToOneThenBefore}, condition (C1) follows for resource~$r$.

      \paragraph{Case 3 ($b_i > b_j$).} In this case, we need to argue that condition (C2) holds for some resource $r^* \in X_j$, that is, that there exists a subset $X^* \subseteq X_j$ with
      \begin{align}\label{eq:cond:Xs}
        \sum_{r \in X^*} \min_{j' \in A \colon b_{j'} < b_i} u_{j'}(\{r\}) < b_i-b_j,
      \end{align}

      such that it holds that  $u_i(X_i) \geq u_i(X_j \setminus (X^* \cup \{r^*\}))$.

      We make a case distinction whether agent $i$ was activated at some point during the algorithm.
      First, suppose  that agent $i$ was activated.
      Let $X^*_k$ be the bundle of agent $k \in A$ after the pick that activated $i$. Note that $b_i > b_j$ implies that $b_i > \min_{i'\in A} b_{i'}$, so by \Cref{alg:rr_init:activeMin}, we have that either $X^*_j=\emptyset$ or it holds that
      \[
        \sum_{r' \in X^*_j \setminus \{r\}} \min_{j' \in A \colon b_{j'} < b_i} u_{j'}(\{r'\}) < b_{i} - b_j
      \]
      for some resource $r \in X^*_j$. Note that it needs to hold that $ X^*_j \neq \emptyset$ since $b_i > b_j$. Thus, let $X^*= X^*_j\setminus \{r\}$ and $r^* = r$. Note that $X^*$ satisfies \Cref{eq:cond:Xs}.
      Next, consider the picking sequence of all picks starting with (including) agent $i$'s first pick. In this sequence, agent $i$ picks the resources in $X_i$ and agent $j$ picks the resource in $X_j \setminus X^*_j$. Since agent $i$ picks first, we have that $u_i(X_i) \geq u_i(X_j \setminus X^*_j) = u_i(X_j \setminus (X^* \cup \{r^*\}))$, so $X^*$ and $r^*$ satisfy (C2).

      Finally, suppose that agent  $i$ was not activated.
      In this case, we will argue that it is always possible to choose $X^*$ and resource $r^*$ in (C2) such that $X^*\cup \{r^*\}=X_j$, which implies that $u_i(X_i) \geq 0=u_i(X_j \setminus (X^* \cup \{r^*\}))$ holds trivially.
      If agent $i$ was not activated, there needs to be an agent $j^* \in A$ with $b_{j^*} < b_i$ such that
      \begin{align}\label{eq:jstar}
        u_{j^*}(X_{j*}) <  b_i - b_{j^*}.
      \end{align}
      If $j=j^*$, this implies that (C2) holds for $X^* = X_{j*}$ and any arbitrary resource $r^* \in X_{j*}$, since
      \[
        b_{i} - b_j  >_{\eqref{eq:jstar}} u_{j^*}(X_{j*}) \geq \sum_{r' \in X_{j^*}} \min_{j' \in A \colon b_{j'} < b_i} u_{j'}(\{r'\}).
      \]
      Thus, it remains to consider the case that $j\neq j^*$. First, we assume that $b_j = b_{j^*}$. Note that in this case, agents $j$ and $j^*$ are active in the same set of rounds $\mathcal{R}$. Thus, using \Cref{alg:rr_init:active}, we get that
      \[
        u_{j^*}(X_{j} \setminus \{r\}) \leq_{\eqref{alg:rr_init:active}} u_{j^*}(X_{j^*}) <_{\eqref{eq:jstar}} b_i -  b_{j^*}=  b_i - b_{j}
      \]
      for some resource $r\in X_j$. Thus, condition (C2) holds for $X^* = X_{j} \setminus \{r\}$ and resource $r^* = r$, since $X^* \cup \{r\} = X_{j}$ and $X^*$ satisfies \Cref{eq:cond:Xs}.
      We consider the remaining case that $b_j \neq b_{j^*}$.

      First, assume that $b_j >  b_{j^*}$. Let $X^*_{j^*}$ be the bundle of agent $j^*$ when agent $j$ was activated. It needs to hold that
      \begin{align}\label{eq:js_reached_j}
        u_{j^*}(X^*_{j^*}) \geq b_j - b_{j^*}.
      \end{align}
      From \Cref{eq:js_reached_j,eq:jstar}, it follows that
      \begin{align}\label{eq:js_after}
        u_{j^*}(X_{j^*} \setminus X^*_{j^*}) < b_i - b_{j}.
      \end{align}
      Let $r \in X_j$ be the first resource picked by agent $j$.
      Consider the picking sequence of all picks starting with the pick after agent $j$'s first pick. In this picking sequence, agent $j$ picks the resources in $X_j \setminus \{r\}$ and agent $j^*$ picks the resources in $X_{j^*} \setminus X^*_{j^*}$.
      Since agent $j^*$ picks before $j$ in this picking sequence, we have that
      \begin{align}\label{eq:jpicks_first}
        u_{j^*}(X_{j^*} \setminus X^*_{j^*}) \geq_{\eqref{alg:rr_init:active}} u_{j^*}(X_j \setminus \{r\}).
      \end{align}
      With this, it follows that
      \begin{align*}
                  b_i - b_{j} >_{\eqref{eq:js_after}}  u_{j^*}(X_{j^*} \setminus X^*_{j^*}) \geq_{\eqref{eq:jpicks_first}} u_{j^*}(X_j \setminus \{r\})\\ \geq \sum_{r' \in X_j \setminus \{r\}} \min_{j' \in A \colon b_{j'} < b_i} u_{j'}(\{r'\}).
      \end{align*}
      This implies that condition (C2) holds for $X^* = X_{j} \setminus \{r\}$ and resource $r^* = r$.

      Finally, assume that  $b_j <  b_{j^*}$. Let $X^*_{j}$ be the bundle of agent $j$ when agent $j^*$ was activated. Note that it cannot be the case that $X^*_{j} = \emptyset$, since $b_j <  b_{j^*}$. Moreover, $b_j <  b_{j^*}$ implies that $b_{j^*} > \min_{i'\in A} b_{i'}$. By \Cref{alg:rr_init:activeMin}, we thus have that there exists a resource $r \in X^*_{j}$ such that
      \begin{align}\label{eq:activeInd}
        b_{j^*} - b_j >  \sum_{r' \in X^*_{j} \setminus \{r\}} \min_{j' \in A \colon b_{j'} < b_{j^*}} u_{j'}(\{r'\}).
      \end{align}
      Consider the picking sequence starting with agent $j^*$'s first pick. In this sequence, agent $j^*$ picks the resources in $X_{j^*}$ and agent $j$ picks the resources in $X_j \setminus X^*_{j}$. Since agent $j^*$ picks first, we get that
      \begin{align}\label{eq:picksJsAfter}
        u_{j^*}(X_j \setminus X^*_{j}) \leq u_{j^*}(X_{j^*}) <_{\eqref{eq:jstar}}  b_i - b_{j^*}.
      \end{align}
      It follows from \Cref{eq:activeInd,eq:picksJsAfter} that
      \begin{align*}
        &b_i - b_j = (b_i - b_{j^*}) + (b_{j^*} - b_j)\\
        >_{\eqref{eq:picksJsAfter},\eqref{eq:activeInd} } \, &u_{j^*}(X_j \setminus X^*_{j}) + \sum_{r' \in X^*_{j} \setminus \{r\}} \min_{j' \in A \colon b_{j'} < b_{j^*}} u_{j'}(\{r'\})\\
        \geq  &\sum_{r' \in X_{j} \setminus \{r\}} \min_{j' \in A \colon b_{j'} < b_{i}} u_{j'}(\{r'\}).
      \end{align*}
      Thus, condition (C2) holds for $X^*=X_j \setminus \{r\}$ and resource $r^* = r$. This completes the proof.
    \end{proof}

    \section{Discussion}
    We initiated the study of fair allocation with initial utilities, focusing on fairness notions that implement the principle of equality of outcome. Our results show that incorporating initial utilities fundamentally alters and complicates the nature of fair allocation problems. Nevertheless, positive algorithmic and axiomatic results can be recovered by tailoring fairness notions to this setting, most notably through our always-satisfiable notion of min-EF1-init, and by focusing on special cases, such as our polynomial-time algorithm for \EFEX{} for identical resources.

    \paragraph{Limitations}
    Our proposed envy notions introduce two key assumptions. First, by including the initial utility of both agents in the comparison to check if a given agent envies another agent, we assume that the initial utilities of all agents are known. Second, differing from standard envy-freeness, this comparison -- and restricting $X^*$ based on the minimum utility in condition (C2) in min-EF1-init -- necessitate interpersonal comparison of utility.

    However, we believe that these assumptions are well-justified in many practically relevant allocation tasks involving initial disparities, where a central decision maker aims to fairly allocate available resources. In such cases, envy-based fairness notions should rather be understood as central solution criteria checked by such a central decision maker rather than as a criterion used by each agent to evaluate the allocation individually. For instance, in the example featured in the introduction, the NGO ARMMAN needs to decide on how to distribute a limited number of available calls among participants to boost program engagement. Here, the NGO, as the central decision maker, has access to (estimates of) the initial engagement of participants from previous data and predictions. Moreover, note that the utilities are \emph{comparable}, since they are quantified by the shared, objective measure of engagement (e.g., time engaged in the program per week) for all agents.
    In such scenarios, our notions offer useful guidance to a central decision maker tasked with achieving fair outcomes.

  \paragraph{Directions for Future Work}
    Our work opens several directions for future research.
    First, it would be interesting to obtain a more fine-grained understanding of the complexity of \EFOEX{}, for instance,
    in the case when resources are identical  (see discussion at the end of \Cref{first-attempt}) or agents value each resource as $0$ or $1$.
    Second, it is worth exploring further relaxations of envy-freeness  for initial utilities.
    Conceptually, min-EF1-init relaxes EF1 in a way that allows more flexibility to assign resources to agents with lower initial utility without creating envy, potentially at the expense of agents with higher initial utility (see \Cref{ex:shortcomings}). It would be worthwhile to study whether always-satisfiable envy notions can also incorporate favorable treatment of agents with higher initial utility or interpolate between both directions. In a similar vein, it would be interesting to investigate the computation of min-EF1-init allocations that are optimal with respect to their utilitarian welfare or the \emph{degree of envy} \citep{DBLP:conf/ijcai/ChevaleyreEEM07}, e.g.,  one could consider an allocation that minimizes the sum of the envy between all pairs of agents according to EF-init or EF1-init. However, computing an optimal solution that is min-EF1-init is NP-hard for both: For the utilitarian welfare, the problem is already NP-hard for EF1 without initial utilities \citep{DBLP:conf/aamas/BarmanG0JN19}; for the degree of envy, finding an optimal solution generalizes the problem of deciding whether an EF(1)-init allocation exists. 
    
    Moreover, studying utility functions beyond additive utilities presents an interesting direction for future work. However, it is easy to show that a complete min-EF1-init allocation may fail to exist for submodular utilities\footnote{Consider an example with two agents $i$ and $j$ with $b_i=0$  and $b_j=10$, and five identical resources. Agent $j$ derives a utility of $10$ from each resource, while agent $i$ has a utility of $10$ only for the first resource allocated to them and a utility of $\epsilon$ for each additional allocated resource for some small $\epsilon>0$. In this instance, according to min-EF1-init, agent $i$ can get at most one resource more than agent $j$. However, this implies that agent $i$ envies $j$ under min-EF1-init in every complete allocation.}, necessitating a further weakening of the solution concept to achieve an always satisfiable notion. Exploring always satisfiable relaxations for submodular utility functions, possibly along the lines of MEF1 \citep{DBLP:journals/teco/CaragiannisKMPS19}, warrants future work.
    
    Furthermore, adapting fairness notions beyond envy-freeness, such as the maximin share (MMS), to equality of outcome deserves attention.\footnote{Note that in Appendix~\ref{app:RW}  we use the work of \citet{DBLP:conf/aaai/HV0V25} to show that  we cannot guarantee the existence of an $\alpha$-approximation of an adapted version of MMS for any $\alpha>0$ in the presence of initial utilities.}
    Lastly, the idea of initial utilities and equality of outcome also deserve attention in other allocation domains, including chore division \citep{DBLP:conf/ijcai/0001CIW19}, online settings \citep{DBLP:journals/jair/KashPS14,DBLP:conf/nips/Benade0022}, and the allocation of divisible resources \citep{brams1996fair}.

    \section*{Acknowledgments}
    We thank an anonymous AAMAS '26 reviewer for the suggestion in Footnote~\ref{fn:suggestion}.

\appendix 

\clearpage

\section*{Appendix}
    Given an instance $I$, we denote the set of all allocations by $$\Pi(I) \coloneqq \left\{ (X_1, \dots, X_n) \in (2^R)^n \,\middle|\, \forall i, j \in A, \, i \neq j \implies X_i \cap X_j = \emptyset  \right\}.$$

    \section{Additional Related Work}\label{app:RW}
    As mentioned in \Cref{sub:rel}, the setting with initial utilities is formally closely related to the completion (or extension) setting considered by \citet{DBLP:conf/aaai/HV0V25} and \citet{DBLP:conf/aaai/DeligkasEGGI25}.
    In their setting, in addition to the resources $R$ and agents $A$ with utility functions, an instance of the completion problem they consider contains a \emph{frozen} (or partial) allocation $\bar{\mathcal{X}}$ of a subset $F\subseteq R$ of the resources (to which they refer as frozen resources). All other resources are called \emph{open}.
    A straightforward way to translate an instance in our setting to an instance in their  setting is the following construction: Given an instance with agents $A$ with utility function $u_i$ and initial utility $b_i$ for $i\in A$ and resources $R$, we construct the following instance in the completion setting, which we call the \emph{derived} instance. The derived instance has the same set of agents $A$. The set of resources is $R'=R\cup \{r^*_{i} \mid i\in A\}$. Agent $i\in A$ has utility function $u'_i : 2^{R'} \to \RRR$ with $u'_i(\{r\})=u_i(\{r\})$ for all $r \in R$ and $u'_i(\{r^*_j\})=b_j$ for all $j \in A$. In the frozen allocation $\bar{\mathcal{X}}$, each agent $i \in A$ is assigned the resource $r^*_i$ and the resources in $R$ are unassigned.

    First, we discuss the work by \citet{DBLP:conf/aaai/DeligkasEGGI25}, who focus on the problem of deciding whether the frozen allocation can be completed to a complete allocation without envy under (standard) EF. As this problem is known to be NP-hard even without any frozen resources, they study the parameterized complexity of this problem with the number of open resources (i.e., $|R\setminus F|$) as parameter. They find that the problem is W[1]-hard for this parameter.
    To obtain a positive result, they consider \emph{agent types} (two agents are of the same type if they have the same utility for each resource), showing that the completion problem is fixed-parameter tractable when parameterized by the number of open resources plus the number of agent types. Moreover, they consider variants of the completion problem where the open resources can be allocated to at most $p$ different agents, and study the parameterized complexity of this problem for parameter $p$.
    While their positive parameterized complexity results could technically be translated to our setting, the parameter regarding the number of open resources becomes vacant, since there is only one frozen resource per agent in the derived instance and all resources from the original instance are open.

    We continue by discussing how our work relates to the work by \citet{DBLP:conf/aaai/HV0V25}. They also consider the completion problem for standard fairness notions, but, in addition to EF, consider EF1 and the maximin share (MMS) \citep{DBLP:conf/nips/CohenEEV24,DBLP:journals/jacm/KurokawaPW18}. Their work too focuses on the computational complexity of determining whether an allocation completing the given frozen allocation exists that satisfies these fairness notions, and they prove that, in contrast to the standard setting, many such problems become computationally intractable. We discuss how their results relate to our setting in the following.
    First, we remark that while the setting with initial utilities can be seen as a special case of the completion setting as described above, an EF1 allocation in the completion setting does not correspond to an EF1-init allocation in the initial utility setting: Since EF1 does not distinguish between open and frozen resources, it is possible to simply ``disregard'' the initial disparities by removing the corresponding frozen resources when comparing two agents bundles to check for envy. This is not possible in our adapted EF1-init notion, which requires that the difference in initial utility is equalized and that the removed resource needs to come from the resources to be distributed (or, that the agent with higher initial utility does not get any resources).
    Therefore, our new envy-based fairness notions for initial utilities are better suited to our specialized setting. Moreover, since our setting is a special case of their completion setting, their negative results do not immediately translate to our setting.

    In addition to EF1, \citet{DBLP:conf/aaai/HV0V25} consider MMS. As a first step towards studying MMS in the initial utility setting, we show that their result that even an arbitrarily bad approximation of MMS cannot be guaranteed extends to our initial utility setting. An intuitive adaptation of MMS to our setting with initial utilities is the following.
    \begin{definition}
      The \emph{max-min-share} of an agent $i\in A$ in instance $I$ with initial utilities is
      \[
        \mu_i \coloneqq \max_{\mathcal{X} \in \Pi(I)} \min_{j \in A} b_j + u_i(X_j).
      \]
      We say that an allocation $\mathcal{X}$ is \emph{max-min-share fair} (MMS-init) if it holds that $b_i + u_i(X_i)\geq \mu_i$ for all $i\in A$.
      An allocation $\mathcal{X}$ is $\alpha$-MMS-init for $0<\alpha\leq 1$, if it holds that $b_i + u_i(X_i)\geq \alpha \cdot \mu_i$ for all $i\in A$.
    \end{definition}
    This definition matches the idea behind the extension of MMS to the completion setting by \citet{DBLP:conf/aaai/HV0V25}, who define the max-min-share as
    \[
      \mu^*_i \coloneqq \max_{\mathcal{X} \in \Pi^*(I)} \min_{j \in A}  u_i(X_j),
    \]
    where $ \Pi^*(I)$ denotes the set of all possible allocations that complete the frozen allocation $\bar{\mathcal{X}}$, that is, $\bar{X}_i\subseteq X_i$ for all $\mathcal{X} \in \Pi^*(I)$ and $i\in A$. If we reduce an allocation instance with initial utilities to the completion setting as described above, then the max-min-share $\mu_i$ in the original instance is equal to the max-min-share $\mu^*_i$ in the derived instance for any $i\in A$: It holds that $u_i(\bar{X}_j)=b_j$ and $\bar{X}_i\subseteq X_i$ for all $\mathcal{X} \in \Pi^*(I)$ and $i,j \in A$. Moreover, let $\mathcal{X}$ be an allocation in the original instance and define allocation $\mathcal{X}'$ with $X'_i = X_i \cup \{r^*_i\}$ for all $i\in A$ in the derived instance. Then, $\mathcal{X}$ is MMS-init if and only if  $\mathcal{X}'$ is MMS, since $b_i + u_i(X_i) = u'_i(X'_i)$ for all $i\in A$.
    Using this fact, it follows from Proposition 1 by \citet{DBLP:conf/aaai/HV0V25} that a Pareto-optimal (PO) and MMS-init allocation always exists for binary additive valuations and initial utilities $b_i\in\{0,1\}$ for all $i\in A$: they show that in their setting such an allocation always exists in this case when the frozen allocation is PO, which is the case in any derived instance.
    \begin{proposition}
      For binary additive valuations and initial utilities $b_i\in\{0,1\}$ for all $i\in A$, an
      MMS-init and PO allocation always exists and
      can be computed in polynomial time.
    \end{proposition}
    On the negative side, \citet{DBLP:conf/aaai/HV0V25} show that for every $0<\alpha\leq 1$, there exists an instance in their completion setting that does not admit an $\alpha$-MMS allocation. Since their construction introduces an asymmetry in the utility of the agents for the frozen resources, we cannot immediately translate their result to our setting, since all agents agree on the initial utilities. However, using an adapted construction, we can prove the result in our setting, using their key idea that each agent $i \in [n]$ believes that the first $i$ agents should share the resources equally, since the remaining agents are already ``rich'' enough. We note that in both constructions, the number of agents grows exponentially in $\alpha^{-1}$.
    \begin{proposition}
      For every $0<\alpha\leq 1$, there exists an instance that does not admit an $\alpha$-MMS-init allocation, even when all resources are identical.
    \end{proposition}
    \begin{proof}
      Let $0<\alpha\leq 1$. For every $\ell\in \NN$, we denote by $H_\ell \coloneqq \sum_{i\in [\ell]} \frac{1}{i}$ the $\ell$-th harmonic number. We choose the number of agents $n$ such that $1<\frac{\alpha}{3}\cdot H_n$. The set $R$ contains $m=2n$ identical resources. Agent $i\in [n]$ has an initial utility and a utility for a resource of $b_i=u_i(\{r\})=m^{(i-1)}$ for all $r\in R$. Observe that for every agent $i\in [n]$, it holds that $\mu_i\geq \floor{\frac{m}{i}} \cdot m^{(i-1)}$, which can be achieved by distributing the resources equally among the first $i$ agents. Note that all remaining agents $j \in A \setminus [i]$ already have at least initial utility $b_j \geq m^i \geq \mu_i$.

      Now, suppose for the sake of contradiction that there exists an $\alpha$-MMS-init allocation $\mathcal{X}$. Then, it needs to hold that $b_i + u_i(X_i)\geq  \alpha \cdot \mu_i \geq \alpha \cdot\floor{\frac{m}{i}} \cdot m^{(i-1)}\geq \frac{\alpha \cdot m}{2i} \cdot m^{(i-1)}$ for all $i\in A$.
      It follows that for all $i\in A$,
      \begin{align*}
        &b_i + u_i(X_i)\geq \frac{\alpha \cdot m}{2i} \cdot m^{(i-1)}\\
        \implies  &u_i(X_i)\geq \frac{\alpha \cdot m}{2i} \cdot m^{(i-1)} - m^{(i-1)}\\
        \implies &|X_i| \geq \frac{\alpha \cdot m}{2i} -1.
      \end{align*}

      Since $\mathcal{X}$ can allocate at most all $m$ resources, it follows that
      \begin{align*}
        & m \geq \sum_{i\in A} (\frac{\alpha \cdot m}{2i} -1) = m(\sum_{i\in A}\frac{\alpha}{2i}) -n\\
        \iff& 1 \geq (\sum_{i\in A}\frac{\alpha}{2i}) - \frac{n}{m} = (\sum_{i\in A}\frac{\alpha}{2i}) - \frac{1}{2}\\
        \iff& \frac{3}{2} \geq \frac{\alpha}{2} \sum_{i\in A}\frac{1}{i} \iff 1 \geq \frac{\alpha}{3} \cdot H_n.
      \end{align*}
      However, this contradicts our choice of $n$, which completes the proof.
    \end{proof}
    

\clearpage


\begin{thebibliography}{49}

\ifx \showCODEN    \undefined \def \showCODEN     #1{\unskip}     \fi
\ifx \showDOI      \undefined \def \showDOI       #1{#1}\fi
\ifx \showISBNx    \undefined \def \showISBNx     #1{\unskip}     \fi
\ifx \showISBNxiii \undefined \def \showISBNxiii  #1{\unskip}     \fi
\ifx \showISSN     \undefined \def \showISSN      #1{\unskip}     \fi
\ifx \showLCCN     \undefined \def \showLCCN      #1{\unskip}     \fi
\ifx \shownote     \undefined \def \shownote      #1{#1}          \fi
\ifx \showarticletitle \undefined \def \showarticletitle #1{#1}   \fi
\ifx \showURL      \undefined \def \showURL       {\relax}        \fi
\providecommand\bibfield[2]{#2}
\providecommand\bibinfo[2]{#2}
\providecommand\natexlab[1]{#1}
\providecommand\showeprint[2][]{arXiv:#2}

\bibitem[\protect\citeauthoryear{Aleksandrov and Walsh}{Aleksandrov and Walsh}{2020}]%
        {DBLP:conf/aaai/AleksandrovW20}
\bibfield{author}{\bibinfo{person}{Martin Aleksandrov} {and} \bibinfo{person}{Toby Walsh}.} \bibinfo{year}{2020}\natexlab{}.
\newblock \showarticletitle{Online fair division: {A} survey}. In \bibinfo{booktitle}{\emph{Proceedings of the 34th {AAAI} Conference on Artificial Intelligence ({AAAI} '20)}}. \bibinfo{publisher}{{AAAI} Press}, \bibinfo{pages}{13557--13562}.
\newblock

\bibitem[\protect\citeauthoryear{Amanatidis, Aziz, Birmpas, Filos{-}Ratsikas, Li, Moulin, Voudouris, and Wu}{Amanatidis et~al\mbox{.}}{2023}]%
        {DBLP:journals/ai/AmanatidisABFLMVW23}
\bibfield{author}{\bibinfo{person}{Georgios Amanatidis}, \bibinfo{person}{Haris Aziz}, \bibinfo{person}{Georgios Birmpas}, \bibinfo{person}{Aris Filos{-}Ratsikas}, \bibinfo{person}{Bo Li}, \bibinfo{person}{Herv{\'{e}} Moulin}, \bibinfo{person}{Alexandros~A. Voudouris}, {and} \bibinfo{person}{Xiaowei Wu}.} \bibinfo{year}{2023}\natexlab{}.
\newblock \showarticletitle{Fair division of indivisible goods: {R}ecent progress and open questions}.
\newblock \bibinfo{journal}{\emph{Artif. Intell.}}  \bibinfo{volume}{322} (\bibinfo{year}{2023}), \bibinfo{pages}{103965}.
\newblock

\bibitem[\protect\citeauthoryear{{ARMMAN}}{{ARMMAN}}{2025}]%
        {armman2025}
\bibfield{author}{\bibinfo{person}{{ARMMAN}}.} \bibinfo{year}{2025}\natexlab{}.
\newblock \bibinfo{title}{Helping mothers and children}.
\newblock \bibinfo{howpublished}{\url{https://armman.org}}.
\newblock
\newblock
\shownote{Accessed: 2025-05-01.}

\bibitem[\protect\citeauthoryear{Aziz, Bir{\'{o}}, Fleiner, Gaspers, de~Haan, Mattei, and Rastegari}{Aziz et~al\mbox{.}}{2022a}]%
        {DBLP:journals/tcs/0001BFGHMR22}
\bibfield{author}{\bibinfo{person}{Haris Aziz}, \bibinfo{person}{P{\'{e}}ter Bir{\'{o}}}, \bibinfo{person}{Tam{\'{a}}s Fleiner}, \bibinfo{person}{Serge Gaspers}, \bibinfo{person}{Ronald de Haan}, \bibinfo{person}{Nicholas Mattei}, {and} \bibinfo{person}{Baharak Rastegari}.} \bibinfo{year}{2022}\natexlab{a}.
\newblock \showarticletitle{Stable matching with uncertain pairwise preferences}.
\newblock \bibinfo{journal}{\emph{Theor. Comput. Sci.}}  \bibinfo{volume}{909} (\bibinfo{year}{2022}), \bibinfo{pages}{1--11}.
\newblock

\bibitem[\protect\citeauthoryear{Aziz, Caragiannis, Igarashi, and Walsh}{Aziz et~al\mbox{.}}{2019}]%
        {DBLP:conf/ijcai/0001CIW19}
\bibfield{author}{\bibinfo{person}{Haris Aziz}, \bibinfo{person}{Ioannis Caragiannis}, \bibinfo{person}{Ayumi Igarashi}, {and} \bibinfo{person}{Toby Walsh}.} \bibinfo{year}{2019}\natexlab{}.
\newblock \showarticletitle{Fair allocation of indivisible goods and chores}. In \bibinfo{booktitle}{\emph{Proceedings of the 28th International Joint Conference on Artificial Intelligence ({IJCAI} '19)}}. \bibinfo{publisher}{ijcai.org}, \bibinfo{pages}{53--59}.
\newblock

\bibitem[\protect\citeauthoryear{Aziz, Li, Moulin, and Wu}{Aziz et~al\mbox{.}}{2022b}]%
        {DBLP:journals/sigecom/AzizLMW22}
\bibfield{author}{\bibinfo{person}{Haris Aziz}, \bibinfo{person}{Bo Li}, \bibinfo{person}{Herv{\'{e}} Moulin}, {and} \bibinfo{person}{Xiaowei Wu}.} \bibinfo{year}{2022}\natexlab{b}.
\newblock \showarticletitle{Algorithmic fair allocation of indivisible items: {A} survey and new questions}.
\newblock \bibinfo{journal}{\emph{SIGecom Exch.}} \bibinfo{volume}{20}, \bibinfo{number}{1} (\bibinfo{year}{2022}), \bibinfo{pages}{24--40}.
\newblock

\bibitem[\protect\citeauthoryear{Barman, Ghalme, Jain, Kulkarni, and Narang}{Barman et~al\mbox{.}}{2019}]%
        {DBLP:conf/aamas/BarmanG0JN19}
\bibfield{author}{\bibinfo{person}{Siddharth Barman}, \bibinfo{person}{Ganesh Ghalme}, \bibinfo{person}{Shweta Jain}, \bibinfo{person}{Pooja Kulkarni}, {and} \bibinfo{person}{Shivika Narang}.} \bibinfo{year}{2019}\natexlab{}.
\newblock \showarticletitle{Fair division of indivisible goods among strategic agents}. In \bibinfo{booktitle}{\emph{Proceedings of the 18th International Conference on Autonomous Agents and MultiAgent Systems ({AAMAS} ’19)}}. \bibinfo{publisher}{IFAAMAS}, \bibinfo{pages}{1811--1813}.
\newblock

\bibitem[\protect\citeauthoryear{Barman, Khan, Shyam, and Sreenivas}{Barman et~al\mbox{.}}{2023}]%
        {DBLP:conf/aaai/Barman0SS23}
\bibfield{author}{\bibinfo{person}{Siddharth Barman}, \bibinfo{person}{Arindam Khan}, \bibinfo{person}{Sudarshan Shyam}, {and} \bibinfo{person}{K.~V.~N. Sreenivas}.} \bibinfo{year}{2023}\natexlab{}.
\newblock \showarticletitle{Finding fair allocations under budget constraints}. In \bibinfo{booktitle}{\emph{Proceedings of the 37th {AAAI} Conference on Artificial Intelligence ({AAAI} ’23)}}. \bibinfo{publisher}{{AAAI} Press}, \bibinfo{pages}{5481--5489}.
\newblock

\bibitem[\protect\citeauthoryear{Barnard and Hepple}{Barnard and Hepple}{2000}]%
        {barnard2000substantive}
\bibfield{author}{\bibinfo{person}{Catherine Barnard} {and} \bibinfo{person}{Bob Hepple}.} \bibinfo{year}{2000}\natexlab{}.
\newblock \showarticletitle{Substantive equality}.
\newblock \bibinfo{journal}{\emph{The Cambridge Law Journal}} \bibinfo{volume}{59}, \bibinfo{number}{3} (\bibinfo{year}{2000}), \bibinfo{pages}{562--585}.
\newblock

\bibitem[\protect\citeauthoryear{Bastani, Drakopoulos, Gupta, Vlachogiannis, Hadjichristodoulou, Lagiou, Magiorkinis, Paraskevis, and Tsiodras}{Bastani et~al\mbox{.}}{2021}]%
        {bastani2021efficient}
\bibfield{author}{\bibinfo{person}{Hamsa Bastani}, \bibinfo{person}{Kimon Drakopoulos}, \bibinfo{person}{Vishal Gupta}, \bibinfo{person}{Ioannis Vlachogiannis}, \bibinfo{person}{Christos Hadjichristodoulou}, \bibinfo{person}{Pagona Lagiou}, \bibinfo{person}{Gkikas Magiorkinis}, \bibinfo{person}{Dimitrios Paraskevis}, {and} \bibinfo{person}{Sotirios Tsiodras}.} \bibinfo{year}{2021}\natexlab{}.
\newblock \showarticletitle{Efficient and targeted COVID-19 border testing via reinforcement learning}.
\newblock \bibinfo{journal}{\emph{Nature}} \bibinfo{volume}{599}, \bibinfo{number}{7883} (\bibinfo{year}{2021}), \bibinfo{pages}{108--113}.
\newblock

\bibitem[\protect\citeauthoryear{Benad{\`{e}}, Halpern, and Psomas}{Benad{\`{e}} et~al\mbox{.}}{2022}]%
        {DBLP:conf/nips/Benade0022}
\bibfield{author}{\bibinfo{person}{Gerdus Benad{\`{e}}}, \bibinfo{person}{Daniel Halpern}, {and} \bibinfo{person}{Alexandros Psomas}.} \bibinfo{year}{2022}\natexlab{}.
\newblock \showarticletitle{Dynamic fair division with partial information}. In \bibinfo{booktitle}{\emph{Proceedings of the 36th Annual Conference on Neural Information Processing Systems (NeurIPS '22)}}.
\newblock

\bibitem[\protect\citeauthoryear{Bentert, Bredereck, Deltl, Jain, and Kellerhals}{Bentert et~al\mbox{.}}{2025}]%
        {ijcai2025p417}
\bibfield{author}{\bibinfo{person}{Matthias Bentert}, \bibinfo{person}{Robert Bredereck}, \bibinfo{person}{Eva Deltl}, \bibinfo{person}{Pallavi Jain}, {and} \bibinfo{person}{Leon Kellerhals}.} \bibinfo{year}{2025}\natexlab{}.
\newblock \showarticletitle{How to Resolve Envy by Adding Goods}. In \bibinfo{booktitle}{\emph{Proceedings of the 34th International Joint Conference on Artificial Intelligence, ({IJCAI' 25})}}. \bibinfo{publisher}{ijcai.org}, \bibinfo{pages}{3753--3761}.
\newblock

\bibitem[\protect\citeauthoryear{Bouveret, Chevaleyre, and Maudet}{Bouveret et~al\mbox{.}}{2016}]%
        {DBLP:reference/choice/BouveretCM16}
\bibfield{author}{\bibinfo{person}{Sylvain Bouveret}, \bibinfo{person}{Yann Chevaleyre}, {and} \bibinfo{person}{Nicolas Maudet}.} \bibinfo{year}{2016}\natexlab{}.
\newblock \showarticletitle{Fair allocation of indivisible goods}.
\newblock In \bibinfo{booktitle}{\emph{Handbook of Computational Social Choice}}, \bibfield{editor}{\bibinfo{person}{Felix Brandt}, \bibinfo{person}{Vincent Conitzer}, \bibinfo{person}{Ulle Endriss}, \bibinfo{person}{J{\'{e}}r{\^{o}}me Lang}, {and} \bibinfo{person}{Ariel~D. Procaccia}} (Eds.). \bibinfo{publisher}{Cambridge University Press}, \bibinfo{pages}{284--310}.
\newblock

\bibitem[\protect\citeauthoryear{Bouveret and Lang}{Bouveret and Lang}{2008}]%
        {DBLP:journals/jair/BouveretL08}
\bibfield{author}{\bibinfo{person}{Sylvain Bouveret} {and} \bibinfo{person}{J{\'{e}}r{\^{o}}me Lang}.} \bibinfo{year}{2008}\natexlab{}.
\newblock \showarticletitle{Efficiency and envy-freeness in fair division of indivisible goods: {L}ogical representation and complexity}.
\newblock \bibinfo{journal}{\emph{J. Artif. Intell. Res.}}  \bibinfo{volume}{32} (\bibinfo{year}{2008}), \bibinfo{pages}{525--564}.
\newblock

\bibitem[\protect\citeauthoryear{Bouveret and Lema{\^{\i}}tre}{Bouveret and Lema{\^{\i}}tre}{2016}]%
        {DBLP:journals/aamas/BouveretL16}
\bibfield{author}{\bibinfo{person}{Sylvain Bouveret} {and} \bibinfo{person}{Michel Lema{\^{\i}}tre}.} \bibinfo{year}{2016}\natexlab{}.
\newblock \showarticletitle{Characterizing conflicts in fair division of indivisible goods using a scale of criteria}.
\newblock \bibinfo{journal}{\emph{Auton. Agents Multi Agent Syst.}} \bibinfo{volume}{30}, \bibinfo{number}{2} (\bibinfo{year}{2016}), \bibinfo{pages}{259--290}.
\newblock

\bibitem[\protect\citeauthoryear{Brams and Taylor}{Brams and Taylor}{1996}]%
        {brams1996fair}
\bibfield{author}{\bibinfo{person}{Steven~J Brams} {and} \bibinfo{person}{Alan~D Taylor}.} \bibinfo{year}{1996}\natexlab{}.
\newblock \bibinfo{booktitle}{\emph{Fair division: {F}rom cake-cutting to dispute resolution}}.
\newblock \bibinfo{publisher}{Cambridge University Press}.
\newblock

\bibitem[\protect\citeauthoryear{Caragiannis, Gravin, and Huang}{Caragiannis et~al\mbox{.}}{2019a}]%
        {DBLP:conf/ec/CaragiannisGH19}
\bibfield{author}{\bibinfo{person}{Ioannis Caragiannis}, \bibinfo{person}{Nick Gravin}, {and} \bibinfo{person}{Xin Huang}.} \bibinfo{year}{2019}\natexlab{a}.
\newblock \showarticletitle{Envy-freeness up to any item with high {N}ash welfare: {T}he virtue of donating items}. In \bibinfo{booktitle}{\emph{Proceedings of the 2019 {ACM} Conference on Economics and Computation ({EC} '19)}}. \bibinfo{publisher}{{ACM}}, \bibinfo{pages}{527--545}.
\newblock

\bibitem[\protect\citeauthoryear{Caragiannis, Kurokawa, Moulin, Procaccia, Shah, and Wang}{Caragiannis et~al\mbox{.}}{2019b}]%
        {DBLP:journals/teco/CaragiannisKMPS19}
\bibfield{author}{\bibinfo{person}{Ioannis Caragiannis}, \bibinfo{person}{David Kurokawa}, \bibinfo{person}{Herv{\'{e}} Moulin}, \bibinfo{person}{Ariel~D. Procaccia}, \bibinfo{person}{Nisarg Shah}, {and} \bibinfo{person}{Junxing Wang}.} \bibinfo{year}{2019}\natexlab{b}.
\newblock \showarticletitle{The unreasonable fairness of maximum {N}ash welfare}.
\newblock \bibinfo{journal}{\emph{{ACM} Trans. Economics and Comput.}} \bibinfo{volume}{7}, \bibinfo{number}{3} (\bibinfo{year}{2019}), \bibinfo{pages}{12:1--12:32}.
\newblock

\bibitem[\protect\citeauthoryear{Chakraborty, Igarashi, Suksompong, and Zick}{Chakraborty et~al\mbox{.}}{2021}]%
        {DBLP:journals/teco/ChakrabortyISZ21}
\bibfield{author}{\bibinfo{person}{Mithun Chakraborty}, \bibinfo{person}{Ayumi Igarashi}, \bibinfo{person}{Warut Suksompong}, {and} \bibinfo{person}{Yair Zick}.} \bibinfo{year}{2021}\natexlab{}.
\newblock \showarticletitle{Weighted envy-freeness in indivisible item allocation}.
\newblock \bibinfo{journal}{\emph{{ACM} Trans. Economics and Comput.}} \bibinfo{volume}{9}, \bibinfo{number}{3} (\bibinfo{year}{2021}), \bibinfo{pages}{18:1--18:39}.
\newblock

\bibitem[\protect\citeauthoryear{Chevaleyre, Endriss, Estivie, and Maudet}{Chevaleyre et~al\mbox{.}}{2007}]%
        {DBLP:conf/ijcai/ChevaleyreEEM07}
\bibfield{author}{\bibinfo{person}{Yann Chevaleyre}, \bibinfo{person}{Ulle Endriss}, \bibinfo{person}{Sylvia Estivie}, {and} \bibinfo{person}{Nicolas Maudet}.} \bibinfo{year}{2007}\natexlab{}.
\newblock \showarticletitle{Reaching envy-free states in distributed negotiation settings}. In \bibinfo{booktitle}{\emph{Proceedings of the 20th International Joint Conference on Artificial Intelligence ({IJCAI} ’07)}}. \bibinfo{publisher}{ijcai.org}, \bibinfo{pages}{1239--1244}.
\newblock

\bibitem[\protect\citeauthoryear{Cohen, Eden, Eden, and Vasilyan}{Cohen et~al\mbox{.}}{2024}]%
        {DBLP:conf/nips/CohenEEV24}
\bibfield{author}{\bibinfo{person}{Ilan~Reuven Cohen}, \bibinfo{person}{Alon Eden}, \bibinfo{person}{Talya Eden}, {and} \bibinfo{person}{Arsen Vasilyan}.} \bibinfo{year}{2024}\natexlab{}.
\newblock \showarticletitle{Plant-and-steal: {T}ruthful fair allocations via predictions}. In \bibinfo{booktitle}{\emph{Proceedings of the 38th Annual Conference on Neural Information Processing Systems (NeurIPS '24)}}.
\newblock

\bibitem[\protect\citeauthoryear{{Commission on Social Determinants of Health}}{{Commission on Social Determinants of Health}}{2008}]%
        {csdh2008}
\bibfield{author}{\bibinfo{person}{{Commission on Social Determinants of Health}}.} \bibinfo{year}{2008}\natexlab{}.
\newblock \bibinfo{booktitle}{\emph{Closing the gap in a generation: {H}ealth equity through action on the social determinants of health}}.
\newblock \bibinfo{publisher}{World Health Organization}, \bibinfo{address}{Geneva}.
\newblock

\bibitem[\protect\citeauthoryear{Deligkas, Eiben, Ganian, Goldsmith, and Ioannidis}{Deligkas et~al\mbox{.}}{2025}]%
        {DBLP:conf/aaai/DeligkasEGGI25}
\bibfield{author}{\bibinfo{person}{Argyrios Deligkas}, \bibinfo{person}{Eduard Eiben}, \bibinfo{person}{Robert Ganian}, \bibinfo{person}{Tiger{-}Lily Goldsmith}, {and} \bibinfo{person}{Stavros~D. Ioannidis}.} \bibinfo{year}{2025}\natexlab{}.
\newblock \showarticletitle{The complexity of extending fair allocations of indivisible goods}. In \bibinfo{booktitle}{\emph{Proceedings of the 39th {AAAI} Conference on Artificial Intelligence ({AAAI} ’25)}}. \bibinfo{publisher}{{AAAI} Press}, \bibinfo{pages}{13745--13753}.
\newblock

\bibitem[\protect\citeauthoryear{Dias, da~Fonseca, de~Figueiredo, and Szwarcfiter}{Dias et~al\mbox{.}}{2003}]%
        {DBLP:journals/tcs/DiasFFS03}
\bibfield{author}{\bibinfo{person}{V{\^{a}}nia M.~F{\'{e}}lix Dias}, \bibinfo{person}{Guilherme~Dias da Fonseca}, \bibinfo{person}{Celina M.~H. de Figueiredo}, {and} \bibinfo{person}{Jayme~Luiz Szwarcfiter}.} \bibinfo{year}{2003}\natexlab{}.
\newblock \showarticletitle{The stable marriage problem with restricted pairs}.
\newblock \bibinfo{journal}{\emph{Theor. Comput. Sci.}} \bibinfo{volume}{306}, \bibinfo{number}{1-3} (\bibinfo{year}{2003}), \bibinfo{pages}{391--405}.
\newblock

\bibitem[\protect\citeauthoryear{Dong, Jong, and King}{Dong et~al\mbox{.}}{2020}]%
        {dong2020does}
\bibfield{author}{\bibinfo{person}{Anmei Dong}, \bibinfo{person}{Morris Siu-Yung Jong}, {and} \bibinfo{person}{Ronnel~B King}.} \bibinfo{year}{2020}\natexlab{}.
\newblock \showarticletitle{How does prior knowledge influence learning engagement? The mediating roles of cognitive load and help-seeking}.
\newblock \bibinfo{journal}{\emph{Frontiers in psychology}}  \bibinfo{volume}{11} (\bibinfo{year}{2020}), \bibinfo{pages}{591203}.
\newblock

\bibitem[\protect\citeauthoryear{Elford}{Elford}{2023}]%
        {sep-equal-opportunity}
\bibfield{author}{\bibinfo{person}{Gideon Elford}.} \bibinfo{year}{2023}\natexlab{}.
\newblock \showarticletitle{Equality of Opportunity}.
\newblock In \bibinfo{booktitle}{\emph{The {Stanford} Encyclopedia of Philosophy} (\bibinfo{edition}{{F}all 2023} ed.)}, \bibfield{editor}{\bibinfo{person}{Edward~N. Zalta} {and} \bibinfo{person}{Uri Nodelman}} (Eds.). \bibinfo{publisher}{Metaphysics Research Lab, Stanford University}.
\newblock

\bibitem[\protect\citeauthoryear{Emanuel and Persad}{Emanuel and Persad}{2023}]%
        {emanuel2023shared}
\bibfield{author}{\bibinfo{person}{Ezekiel~J Emanuel} {and} \bibinfo{person}{Govind Persad}.} \bibinfo{year}{2023}\natexlab{}.
\newblock \showarticletitle{The shared ethical framework to allocate scarce medical resources: {A} lesson from COVID-19}.
\newblock \bibinfo{journal}{\emph{The Lancet}} \bibinfo{volume}{401}, \bibinfo{number}{10391} (\bibinfo{year}{2023}), \bibinfo{pages}{1892--1902}.
\newblock

\bibitem[\protect\citeauthoryear{Farhadi, Ghodsi, Hajiaghayi, Lahaie, Pennock, Seddighin, Seddighin, and Yami}{Farhadi et~al\mbox{.}}{2019}]%
        {DBLP:journals/jair/FarhadiGHLPSSY19}
\bibfield{author}{\bibinfo{person}{Alireza Farhadi}, \bibinfo{person}{Mohammad Ghodsi}, \bibinfo{person}{Mohammad~Taghi Hajiaghayi}, \bibinfo{person}{S{\'{e}}bastien Lahaie}, \bibinfo{person}{David~M. Pennock}, \bibinfo{person}{Masoud Seddighin}, \bibinfo{person}{Saeed Seddighin}, {and} \bibinfo{person}{Hadi Yami}.} \bibinfo{year}{2019}\natexlab{}.
\newblock \showarticletitle{Fair allocation of indivisible goods to asymmetric agents}.
\newblock \bibinfo{journal}{\emph{J. Artif. Intell. Res.}}  \bibinfo{volume}{64} (\bibinfo{year}{2019}), \bibinfo{pages}{1--20}.
\newblock

\bibitem[\protect\citeauthoryear{Garey and Johnson}{Garey and Johnson}{1979}]%
        {gareyJ79}
\bibfield{author}{\bibinfo{person}{Michael~R. Garey} {and} \bibinfo{person}{David~S. Johnson}.} \bibinfo{year}{1979}\natexlab{}.
\newblock \bibinfo{booktitle}{\emph{Computers and intractability: {A} guide to the theory of NP-completeness}}.
\newblock \bibinfo{publisher}{W. H. Freeman}.
\newblock

\bibitem[\protect\citeauthoryear{Goldman and Procaccia}{Goldman and Procaccia}{2014}]%
        {DBLP:journals/sigecom/GoldmanP14}
\bibfield{author}{\bibinfo{person}{Jonathan~R. Goldman} {and} \bibinfo{person}{Ariel~D. Procaccia}.} \bibinfo{year}{2014}\natexlab{}.
\newblock \showarticletitle{Spliddit: {U}nleashing fair division algorithms}.
\newblock \bibinfo{journal}{\emph{SIGecom Exch.}} \bibinfo{volume}{13}, \bibinfo{number}{2} (\bibinfo{year}{2014}), \bibinfo{pages}{41--46}.
\newblock

\bibitem[\protect\citeauthoryear{Halpern and Shah}{Halpern and Shah}{2019}]%
        {DBLP:conf/sagt/HalpernS19}
\bibfield{author}{\bibinfo{person}{Daniel Halpern} {and} \bibinfo{person}{Nisarg Shah}.} \bibinfo{year}{2019}\natexlab{}.
\newblock \showarticletitle{Fair division with subsidy}. In \bibinfo{booktitle}{\emph{Proceedings of the 12th International Symposium on Algorithmic Game Theory ({SAGT} '19)}}. \bibinfo{publisher}{Springer}, \bibinfo{pages}{374--389}.
\newblock

\bibitem[\protect\citeauthoryear{Hardt, Price, and Srebro}{Hardt et~al\mbox{.}}{2016}]%
        {DBLP:conf/nips/HardtPNS16}
\bibfield{author}{\bibinfo{person}{Moritz Hardt}, \bibinfo{person}{Eric Price}, {and} \bibinfo{person}{Nati Srebro}.} \bibinfo{year}{2016}\natexlab{}.
\newblock \showarticletitle{Equality of opportunity in supervised learning}. In \bibinfo{booktitle}{\emph{Proceedings of the 30th Annual Conference on Neural Information Processing Systems (NIPS '16)}}. \bibinfo{pages}{3315--3323}.
\newblock

\bibitem[\protect\citeauthoryear{Hosseini, Sikdar, Vaish, Wang, and Xia}{Hosseini et~al\mbox{.}}{2020}]%
        {hosseiniSVWX20}
\bibfield{author}{\bibinfo{person}{Hadi Hosseini}, \bibinfo{person}{Sujoy Sikdar}, \bibinfo{person}{Rohit Vaish}, \bibinfo{person}{Hejun Wang}, {and} \bibinfo{person}{Lirong Xia}.} \bibinfo{year}{2020}\natexlab{}.
\newblock \showarticletitle{Fair division through information withholding}. In \bibinfo{booktitle}{\emph{Proceedings of the 34th {AAAI} Conference on Artificial Intelligence ({AAAI} ’20)}}. \bibinfo{publisher}{{AAAI} Press}, \bibinfo{pages}{2014--2021}.
\newblock

\bibitem[\protect\citeauthoryear{Hosseini, Sikdar, Vaish, Wang, and Xia}{Hosseini et~al\mbox{.}}{2019}]%
        {DBLP:journals/corr/abs-1907-02583}
\bibfield{author}{\bibinfo{person}{Hadi Hosseini}, \bibinfo{person}{Sujoy Sikdar}, \bibinfo{person}{Rohit Vaish}, \bibinfo{person}{Jun Wang}, {and} \bibinfo{person}{Lirong Xia}.} \bibinfo{year}{2019}\natexlab{}.
\newblock \showarticletitle{Fair division through information withholding}.
\newblock \bibinfo{journal}{\emph{CoRR}}  \bibinfo{volume}{abs/1907.02583} (\bibinfo{year}{2019}).
\newblock
\showeprint[arXiv]{1907.02583}
\urldef\tempurl%
\url{http://arxiv.org/abs/1907.02583}
\showURL{%
\tempurl}

\bibitem[\protect\citeauthoryear{Kash, Procaccia, and Shah}{Kash et~al\mbox{.}}{2014}]%
        {DBLP:journals/jair/KashPS14}
\bibfield{author}{\bibinfo{person}{Ian~A. Kash}, \bibinfo{person}{Ariel~D. Procaccia}, {and} \bibinfo{person}{Nisarg Shah}.} \bibinfo{year}{2014}\natexlab{}.
\newblock \showarticletitle{No agent left behind: {D}ynamic fair division of multiple resources}.
\newblock \bibinfo{journal}{\emph{J. Artif. Intell. Res.}}  \bibinfo{volume}{51} (\bibinfo{year}{2014}), \bibinfo{pages}{579--603}.
\newblock

\bibitem[\protect\citeauthoryear{Kurokawa, Procaccia, and Wang}{Kurokawa et~al\mbox{.}}{2018}]%
        {DBLP:journals/jacm/KurokawaPW18}
\bibfield{author}{\bibinfo{person}{David Kurokawa}, \bibinfo{person}{Ariel~D. Procaccia}, {and} \bibinfo{person}{Junxing Wang}.} \bibinfo{year}{2018}\natexlab{}.
\newblock \showarticletitle{Fair enough: {G}uaranteeing approximate Maximin shares}.
\newblock \bibinfo{journal}{\emph{J. {ACM}}} \bibinfo{volume}{65}, \bibinfo{number}{2} (\bibinfo{year}{2018}), \bibinfo{pages}{8:1--8:27}.
\newblock

\bibitem[\protect\citeauthoryear{Lane, Sarkies, Martin, and Haines}{Lane et~al\mbox{.}}{2017}]%
        {LANE201711}
\bibfield{author}{\bibinfo{person}{Haylee Lane}, \bibinfo{person}{Mitchell Sarkies}, \bibinfo{person}{Jennifer Martin}, {and} \bibinfo{person}{Terry Haines}.} \bibinfo{year}{2017}\natexlab{}.
\newblock \showarticletitle{Equity in healthcare resource allocation decision making: {A} systematic review}.
\newblock \bibinfo{journal}{\emph{Social Science \& Medicine}}  \bibinfo{volume}{175} (\bibinfo{year}{2017}), \bibinfo{pages}{11--27}.
\newblock

\bibitem[\protect\citeauthoryear{Lipton, Markakis, Mossel, and Saberi}{Lipton et~al\mbox{.}}{2004}]%
        {DBLP:conf/sigecom/LiptonMMS04}
\bibfield{author}{\bibinfo{person}{Richard~J. Lipton}, \bibinfo{person}{Evangelos Markakis}, \bibinfo{person}{Elchanan Mossel}, {and} \bibinfo{person}{Amin Saberi}.} \bibinfo{year}{2004}\natexlab{}.
\newblock \showarticletitle{On approximately fair allocations of indivisible goods}. In \bibinfo{booktitle}{\emph{Proceedings of the 5th {ACM} Conference on Electronic Commerce (EC '04)}}. \bibinfo{publisher}{{ACM}}, \bibinfo{pages}{125--131}.
\newblock

\bibitem[\protect\citeauthoryear{Liu, Lu, Suzuki, and Walsh}{Liu et~al\mbox{.}}{2024}]%
        {DBLP:journals/jair/LiuLSW24}
\bibfield{author}{\bibinfo{person}{Shengxin Liu}, \bibinfo{person}{Xinhang Lu}, \bibinfo{person}{Mashbat Suzuki}, {and} \bibinfo{person}{Toby Walsh}.} \bibinfo{year}{2024}\natexlab{}.
\newblock \showarticletitle{Mixed fair division: {A} survey}.
\newblock \bibinfo{journal}{\emph{J. Artif. Intell. Res.}}  \bibinfo{volume}{80} (\bibinfo{year}{2024}), \bibinfo{pages}{1373--1406}.
\newblock

\bibitem[\protect\citeauthoryear{Moulin}{Moulin}{2004}]%
        {moulin2004fair}
\bibfield{author}{\bibinfo{person}{Herv{\'e} Moulin}.} \bibinfo{year}{2004}\natexlab{}.
\newblock \bibinfo{booktitle}{\emph{Fair division and collective welfare}}.
\newblock \bibinfo{publisher}{MIT press}.
\newblock

\bibitem[\protect\citeauthoryear{Moulin}{Moulin}{2019}]%
        {moulin2019fair}
\bibfield{author}{\bibinfo{person}{Herv{\'e} Moulin}.} \bibinfo{year}{2019}\natexlab{}.
\newblock \showarticletitle{Fair division in the internet age}.
\newblock \bibinfo{journal}{\emph{Annual Review of Economics}} \bibinfo{volume}{11}, \bibinfo{number}{1} (\bibinfo{year}{2019}), \bibinfo{pages}{407--441}.
\newblock

\bibitem[\protect\citeauthoryear{{Prakash HV}, Igarashi, and Vaish}{{Prakash HV} et~al\mbox{.}}{2025}]%
        {DBLP:conf/aaai/HV0V25}
\bibfield{author}{\bibinfo{person}{Vishwa {Prakash HV}}, \bibinfo{person}{Ayumi Igarashi}, {and} \bibinfo{person}{Rohit Vaish}.} \bibinfo{year}{2025}\natexlab{}.
\newblock \showarticletitle{Fair and efficient completion of indivisible goods}. In \bibinfo{booktitle}{\emph{Proceedings of the 39th {AAAI} Conference on Artificial Intelligence ({AAAI} ’25)}}. \bibinfo{publisher}{{AAAI} Press}, \bibinfo{pages}{14045--14053}.
\newblock

\bibitem[\protect\citeauthoryear{{Scottish Government}}{{Scottish Government}}{2021}]%
        {scottishgov2021fairer}
\bibfield{author}{\bibinfo{person}{{Scottish Government}}.} \bibinfo{year}{2021}\natexlab{}.
\newblock \bibinfo{title}{Fairer Scotland Duty: {G}uidance for public bodies}.
\newblock
\newblock

\bibitem[\protect\citeauthoryear{Trump}{Trump}{2025}]%
        {trump2025equality}
\bibfield{author}{\bibinfo{person}{Donald~J. Trump}.} \bibinfo{year}{2025}\natexlab{}.
\newblock \bibinfo{title}{Restoring equality of opportunity and meritocracy}.
\newblock
\newblock
\newblock
\shownote{Executive Order 14281.}

\bibitem[\protect\citeauthoryear{Verma, Singh, Mate, Verma, Gorantla, Madhiwalla, Hegde, Thakkar, Jain, Tambe, and Taneja}{Verma et~al\mbox{.}}{2023}]%
        {DBLP:journals/aim/VermaSMVGMHTJTT23}
\bibfield{author}{\bibinfo{person}{Shresth Verma}, \bibinfo{person}{Gargi Singh}, \bibinfo{person}{Aditya Mate}, \bibinfo{person}{Paritosh Verma}, \bibinfo{person}{Sruthi Gorantla}, \bibinfo{person}{Neha Madhiwalla}, \bibinfo{person}{Aparna Hegde}, \bibinfo{person}{Divy Thakkar}, \bibinfo{person}{Manish Jain}, \bibinfo{person}{Milind Tambe}, {and} \bibinfo{person}{Aparna Taneja}.} \bibinfo{year}{2023}\natexlab{}.
\newblock \showarticletitle{Expanding impact of mobile health programs: {SAHELI} for maternal and child care}.
\newblock \bibinfo{journal}{\emph{{AI} Mag.}} \bibinfo{volume}{44}, \bibinfo{number}{4} (\bibinfo{year}{2023}), \bibinfo{pages}{363--376}.
\newblock

\bibitem[\protect\citeauthoryear{Verma, Zhao, Shah, Boehmer, Taneja, and Tambe}{Verma et~al\mbox{.}}{2024}]%
        {DBLP:conf/uai/VermaZSBTT24}
\bibfield{author}{\bibinfo{person}{Shresth Verma}, \bibinfo{person}{Yunfan Zhao}, \bibinfo{person}{Sanket Shah}, \bibinfo{person}{Niclas Boehmer}, \bibinfo{person}{Aparna Taneja}, {and} \bibinfo{person}{Milind Tambe}.} \bibinfo{year}{2024}\natexlab{}.
\newblock \showarticletitle{Group Fairness in predict-then-optimize settings for restless bandits}. In \bibinfo{booktitle}{\emph{Proceedings of the 40th Conference on Uncertainty in Artificial Intelligence (UAI '24)}} \emph{(\bibinfo{series}{Proceedings of Machine Learning Research}, Vol.~\bibinfo{volume}{244})}. \bibinfo{publisher}{{PMLR}}, \bibinfo{pages}{3448--3469}.
\newblock

\bibitem[\protect\citeauthoryear{Wu, Li, and Gan}{Wu et~al\mbox{.}}{2021}]%
        {DBLP:journals/corr/abs-2012-03766}
\bibfield{author}{\bibinfo{person}{Xiaowei Wu}, \bibinfo{person}{Bo Li}, {and} \bibinfo{person}{Jiarui Gan}.} \bibinfo{year}{2021}\natexlab{}.
\newblock \showarticletitle{Budget-feasible maximum {N}ash social welfare is almost envy-free}. In \bibinfo{booktitle}{\emph{Proceedings of the 30th International Joint Conference on Artificial Intelligence ({IJCAI} ’21)}}. \bibinfo{publisher}{ijcai.org}, \bibinfo{pages}{465--471}.
\newblock

\bibitem[\protect\citeauthoryear{Xia and Conitzer}{Xia and Conitzer}{2011}]%
        {DBLP:journals/jair/XiaC11}
\bibfield{author}{\bibinfo{person}{Lirong Xia} {and} \bibinfo{person}{Vincent Conitzer}.} \bibinfo{year}{2011}\natexlab{}.
\newblock \showarticletitle{Determining possible and necessary winners given partial orders}.
\newblock \bibinfo{journal}{\emph{J. Artif. Intell. Res.}}  \bibinfo{volume}{41} (\bibinfo{year}{2011}), \bibinfo{pages}{25--67}.
\newblock

\bibitem[\protect\citeauthoryear{Zhou, Bai, and Wu}{Zhou et~al\mbox{.}}{2023}]%
        {DBLP:conf/icml/0002B023}
\bibfield{author}{\bibinfo{person}{Shengwei Zhou}, \bibinfo{person}{Rufan Bai}, {and} \bibinfo{person}{Xiaowei Wu}.} \bibinfo{year}{2023}\natexlab{}.
\newblock \showarticletitle{Multi-agent online scheduling: {MMS} allocations for indivisible items}. In \bibinfo{booktitle}{\emph{Proceedings of the 40th International Conference on Machine Learning ({ICML} '23)}} \emph{(\bibinfo{series}{Proceedings of Machine Learning Research}, Vol.~\bibinfo{volume}{202})}. \bibinfo{publisher}{{PMLR}}, \bibinfo{pages}{42506--42516}.
\newblock

\end{thebibliography}
\end{document}